\documentclass[aps,pra,twocolumn,showpacs,groupedaddress,notitlepage,superscriptaddress,floatfix,nofootinbib,longbibliography]{revtex4-2}
\usepackage{placeins}
\usepackage{graphicx}
\usepackage{subfigure}
\usepackage{amsmath}

\usepackage{physics}
\usepackage{amsthm}
\usepackage{bm}
\usepackage{amsfonts}

\newcommand {\be}{\begin{equation}}
\newcommand {\ee}{\end{equation}}
\newcommand{\beq}{\begin{equation}}
\newcommand{\eeq}{\end{equation}}
\newcommand{\beqnn}{\begin{equation*}}
\newcommand{\eeqnn}{\end{equation*}}
\newcommand{\bea}{\begin{eqnarray}}
\newcommand{\eea}{\end{eqnarray}}
\newcommand{\ba}{\begin{eqnarray}}
\newcommand{\ea}{\end{eqnarray}}
\newcommand{\beann}{\begin{eqnarray*}}
\newcommand{\eeann}{\end{eqnarray*}}
\newcommand{\bes} {\begin{subequations}}
\newcommand{\ees} {\end{subequations}}

\newcommand {\nn}{\notag}
\newcommand{\mcV}{\mathcal{V}}
\newcommand{\mcE}{\mathcal{E}}

\graphicspath{ {./images/} }

\newtheorem{mytheorem}{Theorem}
\newtheorem{mylemma}{Lemma}
\newtheorem{mycorollary}{Corollary}

\usepackage{hyperref}
\hypersetup{
    colorlinks=true, 
    linkcolor=cyan,
    citecolor=magenta, 
    filecolor=magenta, 
    urlcolor=cyan,
    runcolor=cyan
}

\newtheorem{definition}{Definition}

\usepackage{tikz}
\usetikzlibrary{quantikz}

\begin{document}

\title{Low overhead universality and quantum supremacy using only $Z$-control}

\author{Brian Barch}
\affiliation{Department of Physics and Astronomy, University of Southern California, Los Angeles, California 90089, USA}
\affiliation{Center for Quantum Information Science \& Technology, University of Southern California, Los Angeles, California 90089, USA}
\author{Razieh Mohseninia}
\affiliation{Center for Quantum Information Science \& Technology, University of Southern California, Los Angeles, California 90089, USA}
\author{Daniel Lidar}
\affiliation{Department of Physics and Astronomy, University of Southern California, Los Angeles, California 90089, USA}
\affiliation{Center for Quantum Information Science \& Technology, University of Southern California, Los Angeles, California 90089, USA}
\affiliation{Department of Electrical and Computer Engineering, University of Southern California, Los Angeles, California 90089, USA}
\affiliation{Department of Chemistry, University of Southern California, Los Angeles, California 90089, USA}

\begin{abstract}
We consider a model of quantum computation we call ``Varying-$Z$" (V$Z$),  defined by applying controllable $Z$-diagonal Hamiltonians in the presence of a uniform and constant external $X$-field,  
and prove that it is universal, even in 1D. Universality is demonstrated by construction of a universal gate set with $O(1)$ depth overhead. We then use this construction to describe a circuit whose output distribution cannot be classically simulated unless the polynomial hierarchy collapses, with the goal of providing a low-resource method of demonstrating quantum supremacy. The V$Z$ model can achieve quantum supremacy in $O(n)$ depth, equivalent to the random circuit sampling models despite a higher degree of homogeneity: it requires no individually addressed $X$-control. 
\end{abstract}

\maketitle

\section{Introduction}

In the current era of noisy intermediate scale quantum computers~\cite{Preskill:2018aa}, 
quantum architectures are limited by connectivity, gate fidelity, and various other sources of errors that limit both circuit depth and width. 
In response, various models of quantum computation have been developed that are designed to be 
relatively easy to implement on existing hardware.
The strength of these models is then confirmed by demonstrating their ability to achieve universality~\cite{Deutsch:85,DiVincenzo1995,Lloyd:95,lloyd,raussendorf} 
or quantum supremacy~\cite{Preskill:2012aa,bremner1,farhi,aaronson2016,Harrow:2017aa,bouland2018,bermejo,haferkamp,Arute:2019aa}. Universality is a stronger attribute, as it implies the ability to reproduce the quantum supremacy results of other models.

In this spirit, here we propose a model of quantum computation that is computationally universal even when restricted to a one-dimensional (1D) chain of qubits with only nearest-neighbor interactions and a limited degree of control. This ``Varying-$Z$" (V$Z$) model is defined by applying a series of $Z$-diagonal Hamiltonians in the presence of a constant and homogeneous (i.e., qubit-permutation-invariant) $X$-field requiring no individually addressed control. We consider the V$Z$ model both on a general graph and in 1D. The latter is theoretically motivated by the question of the quantum computational power of 1D systems~\cite{Aharonov:2009tm,raussendorf} and experiments with 1D systems, such as chains of fluxonium qubits~\cite{Meier:2015aa}, and chains of transmons which were used in a prequel to Google's quantum supremacy work~\cite{Neill:2018aa}.
%\footnote{Note that while trapped polar molecules~\cite{DeMille:2002aa}, trapped ions in a linear Paul trap~\cite{Cirac:95}, and condensed phase analogs thereof~\cite{Brown:2001aa}, are by design physically 1D systems, they can exhibit long-range interactions and hence are not geometrically 1D.} 
The general V$Z$ model is physically motivated by physical systems subject to always-on transverse fields, 
such as superconducting flux-qubit architectures~\cite{You:2007aa,harris_flux_qubit_2010,Yan:2016aa,grover2020fast} which experience a small but always-on $X$-field in a quantum annealing context~\cite{mozgunov2020quantum}.
% that also reduces decoherence (see Appendix~\ref{app:flux}). 
%Other systems with an always-on field along a particular axis are common, e.g., liquid-state nuclear magnetic resonance~\cite{Gershenfeld:97}, or nitrogen-vacancy centers in diamond~\cite{Sar:2012km}.

The outline of this paper is as follows: in Sec.~\ref{sec:review} we review the previous results on the universality and supremacy of 1D models. In Sec.~\ref{sec:model} we define our model and in Sec.~\ref{sec:proof} we demonstrate universality by reconstructing a universal gate set. In Sec.~\ref{sec:supremacy} we apply this universality result to generate distributions known to demonstrate quantum supremacy, and in Sec.~\ref{sec:conclusion} we close with concluding remarks. Additional technical details are provided in the appendix.

\section{Background: universality and supremacy via 1D models} 
\label{sec:review}

\begin{table*}
    \begin{tabular}{|c|c|c|c|c|c|}
    \hline
         Model & Inhomogeneity & Physical Qubits & UQC Runtime Scaling& QS Runtime Scaling& Reference \\
         \hline
         QAOA & Control Unit & $O(nlog(n))$ &  $O(n^5d^5\epsilon^{-4})$ & $O(n^{10}\epsilon^{-4})$ & \cite{lloyd}\\
         Quantum Circuits & - & $4n+2$ & $O(nd)$ & $O(n^2)$& \cite{raussendorf}\\
         V$Z$ Model & $Z$ interactions & $n$ & $O(d)$ & $O(n)$ & This work\\
         Quantum Circuits & $X$,$Z$ interactions & $n$ & $d$& $O(n)$ & \cite{bermejo, haferkamp}\\
         \hline
    \end{tabular}
    \caption{Table of various 1D models of universal quantum computation. Models are compared in terms of their requirements needed to reproduce the output distribution of a 1D circuit of $n$ logical qubits and depth $d$, to within total variation distance $\epsilon$. The UQC and QS columns give the runtime required for universality and supremacy, respectively.}
    \label{tab:1}
\end{table*}

Universal Quantum Computation (UQC) is, informally, the ability to solve any problem that can be solved by any quantum computer, or more formally, the ability to approximate any unitary transformation to arbitrary accuracy in polynomial runtime~\cite{Deutsch:85,DiVincenzo1995}. Quantum Supremacy (QS) is, also informally, the ability to solve problems that cannot be solved in the same amount of time by any classical computer~\cite{Preskill:2012aa}. More formally, QS, is the ability to generate a probability distribution that cannot be efficiently approximated to arbitrary accuracy by any classical computer with access to randomness unless the polynomial hierarchy collapses, which is believed to be unlikely~\cite{aaronson2016,Harrow:2017aa}. A number of problems exist that are known to be solvable in theory by universal quantum computers but not classical computers, so a quantum computer's universality implies its quantum supremacy~\cite{bremner1,farhi}. A brief further discussion of the complexity basis of supremacy is provided in Appendix~\ref{app:complex}.

In recent years a number of alternative models of quantum computation have been proposed, with a variety of dimensionality, circuit depth, and homogeneity requirements needed to achieve universality. Here we restrict to considering models that are universal in 1D, and compare the models on their requirements needed to reproduce a depth $d$ universal gate set (UGS) based quantum circuit, as summarized in Table~\ref{tab:1}. Let us now explain the gain achieved by the V$Z$ model, by contrasting it with the other models featured in this table. 

It is known that 1D gate-based quantum circuits can achieve QS in $O(n)$ depth \cite{bermejo, haferkamp}, so we can use 1D quantum circuits as a reference point to compare the runtime requirements of other universal models in achieving QS. One such model is the Quantum Approximate Optimization Algorithm (QAOA)~\cite{farhi2014quantum} (briefly reviewed in Appendix~\ref{app:notQAOA}) equipped with Broadcast Quantum Cellular Automata (BQCA)~\cite{Simon}, that was shown to be universal~\cite{lloyd,morales}. QAOA defined in 1D consists of a chain of qubits which undergo evolution that alternates between a homogeneous $X$-field and a potentially inhomogeneous $Z$-diagonal Hamiltonian.
 BQCA requires addressed control of a single qubit, the control unit, which it uses to break translational symmetry and reproduce local gates on other qubits in the chain. Using QAOA with BQCA to reproduce the output distribution of a given 1D depth $d$ quantum circuit to within total variation (Kolmogorov) distance $\epsilon$ requires a runtime of $O(n^5 d^5 \epsilon^{-4})$ in the worst case, as shown in Appendix~\ref{app:lloyd}. 

The QAOA model requires only a fixed $Z$-diagonal Hamiltonian repeated 
%across all $p$ individual layers (Def.~\ref{def:1}), 
with different evolution times.
%$\gamma_k$ (see Appendix~\ref{app:notQAOA}). 
Alternatively, if one is capable of implementing each desired gate as an alternating sequence of homogeneous nearest neighbor entangling gates and homogeneous local rotations, using boundary conditions of the underlying architecture to introduce spatial control, universality can be achieved in depth $O(nd)$ via the model of Ref.~\cite{raussendorf}. This model can reproduce the target circuit exactly in the absence of noise, so the time cost is independent of $\epsilon$. Compared with QAOA, this model works by applying a set of more general yet still homogeneous quantum gates.

Resource-wise, the 1D V$Z$ model defined here can be thought of as a midpoint between the homogenous circuit model of \cite{raussendorf} and general 1D quantum circuits, in that it only requires individually addressed control of the $Z$-interactions. Likewise, the asymptotic depth requirement of $O(d)$ to achieve universality is between that of \cite{raussendorf} and the original UGS-based universal quantum circuit being simulated.

While our primary concern is with universality of the V$Z$ model, in Sec.~\ref{sec:supremacy} we also provide an example of a problem not contained within the complexity class BPP, which could be used to demonstrate quantum supremacy in a practical setting, e.g., using trapped ions or flux qubits.

\section{The Varying-$Z$ Model}
\label{sec:model}

We will analyze the V$Z$ model from the perspective of gate layers, rather than individual gates.

\begin{definition}[Gate Layer]
A gate layer is a depth-$1$ operation, equivalent to a set of commuting gates applied in parallel in the circuit model.
\label{def:1}
\end{definition}

The V$Z$ model reproduces a gate layer from a circuit in the gate model using a series of 
\emph{applied} layers: 
gate layers corresponding to the application of a single time-independent Hamiltonian, which are natural to the V$Z$ model. We refer to the gate layer being reproduced from applied layers as the \textit{effective} layer.
In the V$Z$ model all $n$ qubits are initially prepared in the $\ket{+}^{\otimes n}$ state and then acted upon by a series of applied unitary layers. These unitaries are generated by a series of Hamiltonians composed of two terms: a $Z$-diagonal term $H_l^z$ which varies by applied layer $l$, and a constant, homogeneous $X$-field $H^x$ which is independent of the layer. Each applied layer $l$ is applied for time $t_l$. Note that $l$ plays the role of a discrete time index.
We take the $Z$-Hamiltonians to be two-local between neighboring qubits located on the vertices $i\in \mathcal{V}$ of some underlying graph $(\mathcal{V},\mathcal{E})$, and with interactions $w_{ij}$ on the edges $(i,j)\in\mathcal{E}$ that are uniform in magnitude but can be turned on or off by edge. The $l^{\text{th}}$ Hamiltonian may be written as:
\bes
\label{eq:H_l}
\begin{align}
    H_l &= H_l^z+H^x \\
    H^x &= a \sum_{i\in \mathcal{V}} X_i\\ 
    H_l^z &= b_l \sum_{(i,j)\in\mathcal{E}} w_{l,ij} Z_i Z_{j} + c_l \sum_{i\in \mathcal{V}} v_{l,i} Z_i  ,
\end{align}
\ees
where $w_{l,ij} \in \{0,1\}$ and $v_{l,i}\in \{0,1\}$ respectively switch the interactions and local fields on or off for the $l$th applied layer, and $a>0$ is fixed throughout the circuit. The total number of qubits is $n=|\mathcal{V}|$.

\begin{definition}[Varying-$Z$ model]
Starting from the initial state $\ket{\boldsymbol{+}}\equiv\ket{+}^{\otimes n}$, apply each Hamiltonian $H_l$ [Eq.~\eqref{eq:H_l}] for corresponding time $t_l$, measure all the qubits in the $Z$-basis, and sample the final state.
\label{def:TZ}
\end{definition}

The output probability distribution is given by:
\begin{equation}
    P(\mathbf s) = |\langle \mathbf s| \prod_l e^{-it_l H_l}\ket{\boldsymbol{+}}|^2 ,
\end{equation}
where $\mathbf s \in \{0,1\}^n$. 

Our main result is the following:
\begin{mytheorem}
\label{thm:1}
The V$Z$ model can simulate an arbitrary depth $d$ quantum circuit on a graph of maximum degree $\Delta$ to arbitrary accuracy in depth $O(d\Delta)$ on the same underlying graph.
\end{mytheorem}

Notably, any circuit on a graph of bounded degree, e.g., a 1D chain with $\Delta=2$, may be simulated by the V$Z$ model in depth $O(d)$. Given the universality of 1D quantum circuits, this has the immediate consequence:

\begin{mycorollary}
\label{cor:1D}
The 1D V$Z$ model is quantum computationally universal.
\end{mycorollary}

Note that if $t_l$ or $a$ had a sufficient degree of inhomogeneity (i.e., dependence on the qubit index $i$), or if $a$ were allowed to vary by layer and vanish, we would have sufficient control to directly construct the single qubit UGS $\{W, T\}$, where $W=ie^{-i\frac{\pi}{2\sqrt{2}}(X+Z)}$ is the Hadamard gate and $T=e^{-i\frac{\pi}{8}Z}$ is the $T$-gate. However, Theorem~\ref{thm:1} shows that on graphs of bounded degree and perhaps in general, neither of these relaxations provides a benefit over the already asymptotically optimal depth scaling. In other words, replacing $X$-field control with a constant, always-on transverse field $H^x$ is sufficient for low-overhead universality, as long as the $Z$-diagonal Hamiltonian can be updated between successive applied gate layers. This is clearly a significant  simplification in terms of control requirements over the standard UGS approach.

Note that the initial state $\ket{\boldsymbol{+}}$ is the ground state of $-H^x$, so it can be prepared by turning on this Hamiltonian and waiting for the system to relax into its ground state. It can also be prepared starting from the $\ket{0}^{\otimes n}$ state and applying the global Hadamard gate $W^{\otimes n}$, which is compatible with the V$Z$ model since it requires no inhomogeneity of the $X$-field.

We remark that the V$Z$ model resembles QAOA~\cite{farhi2014quantum} to some extent. The main differences are 
the fact that in the V$Z$ model the $X$-field is always on (whereas in QAOA one alternates between $H_l^z$ and $H^x$), and 
that in V$Z$ model we assume that the $b$ and $c$ coefficients are $l$-dependent, whereas in QAOA they may vary by qubit but not by $l$ (see also Appendix~\ref{app:notQAOA}). 

\textit{Proof outline of Theorem~\ref{thm:1}.}
Two-qubit gates are universal for quantum computation~\cite{DiVincenzo1995}, and an arbitrary two-qubit gate can be produced with a constant number of single qubit unitaries and gates generated by $ZZ$ interactions~\cite{Barenco1995,nielsen2010quantum}, which we refer to as $ZZ$-gates. Thus, in order to reproduce an arbitrary quantum circuit, it is sufficient to demonstrate the ability to generate arbitrary single qubit unitaries and $ZZ$-gates. We will first demonstrate the ability to reproduce layers corresponding to arbitrary single qubit unitaries (Lemma~\ref{lem:single}) and then corresponding to $ZZ$-gates with arbitrary real coupling constants (Lemma~\ref{lem:couple}). The technical challenge in proving these results is to deal with the fact the $X$-field is always on. 

Each type of effective gate layer $G_L$ can be implemented using $l_{\max}(L)\in O(1)$ applied layers of the V$Z$ model, corresponding to the decomposition
\begin{equation}
    G_L = \prod_{l=1}^{l_{\max}(L)} e^{-it_l(H^x + H_l^z)} .
\end{equation}

As each effective layer can apply gates across all qubits in parallel, we will then analyze the decomposition of an arbitrary circuit into layers based on its UGS, and conclude that the circuit can be implemented in the V$Z$ model with depth overhead proportional to the number of gates that may act on the same qubit within a single gate layer.

\section{Proof of universality of the V$Z$ model}
\label{sec:proof}

In this section we provide a detailed proof of Theorem~\ref{thm:1}. 

\subsection{Single-qubit gate layers}
Consider an effective gate layer $G_L$ corresponding to identical arbitrary single qubit unitaries $g_i$, $i\in\mathcal{V}$. Let $g_i$ apply a rotation by some angle $\gamma$ about some axis $\vec r = (\sin(2\theta)\cos(2\phi),\ \sin(2\theta)\sin(2\phi),\ \cos(2\theta))$ of the Bloch sphere. This layer may then be decomposed as
\beq
\label{eq:logicalsinglelayer}
   G_L  = \bigotimes_{i} (g_i)^{v_i}= e^{-i\gamma \sum_i v_i \vec r \cdot \vec \sigma_i} \ , \quad v_i \in \{0,1\} 
\eeq
for $\vec \sigma_i = (X_i, Y_i, Z_i)$, and $v_i = 1$ iff unitary $g_i$ is applied to qubit $i$. How would we implement $G_L$ using just the components of the V$Z$ model? It would appear that a simple Euler angles construction should suffice, but we explain in Appendix~\ref{app:Euler} why this approach fails. Instead, the following lemma provides the answer: 

\begin{mylemma}\label{lem:single}
The V$Z$ model can implement an arbitrary effective single-qubit gate layer $G_L$ [Eq.~\eqref{eq:logicalsinglelayer}] in three applied layers.
\end{mylemma} 

\begin{proof}
We will show that $G_L$ may be decomposed into a product of three applied unitary layers as

\be
\label{eq:VUV}
    G_L = V^{\otimes n} \bigotimes_{i=1}^n U_i \ (V^{\dag})^ {\otimes n} = \bigotimes_{i\in \mcV} V_i U_i V^\dag_i ,
    \ee
with
\begin{subequations}
\label{eq:UiV}	 
\begin{align}
\label{eq:Ui}
    U_i &= e^{-it(a X_i + c v_i Z_i)}\\ 
 \label{eq:V}
 V_i &= e^{-it'(aX_i+c' Z_i)} .
\end{align}
 \end{subequations}    
$V^\dag$ may be implemented modulo $\pi$, and does not require changing the $X$-field strength $a$ (we suppress the $i$ subscript where convenient). The resulting effective and applied layers are depicted in Fig.~\ref{fig:singlegatelayer}.

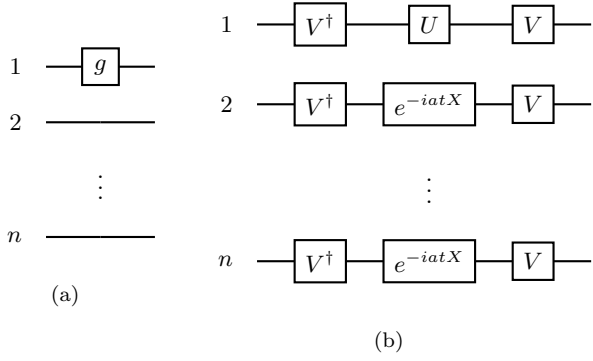
\begin{figure}
\subfigure[\ ]{\begin{quantikz}
    \lstick{1}\hspace{2mm} &\gate{g}&\qw\\ 
    \lstick{2}\hspace{2mm} &\qw&\qw\\
    &\vdots\\
    \lstick{$n$}\hspace{2mm} &\qw&\qw\\
\end{quantikz}}\hspace{5mm}
\subfigure[\ ]{\begin{quantikz}
    \lstick{1}\hspace{2mm} &\gate{V^\dag}&\gate{U}&\gate{V}&\qw\\
    \lstick{2}\hspace{2mm} &\gate{V^\dag}&\gate{e^{-iatX}}&\gate{V}&\qw\\
    &&\vdots\\
    \lstick{$n$}\hspace{2mm} &\gate{V^\dag}&\gate{e^{-iatX}}&\gate{V}&\qw\\
\end{quantikz}}
\caption{a) Example effective single-qubit gate layer $G_L$ in 1D, with a gate $g$ acting on only the first qubit. b) Implementation as applied layers in the V$Z$ model. The time $t$ is chosen to make the effect on qubits $\{2,\dots,n\}$ the identity gate.}
\label{fig:singlegatelayer}
\end{figure}

Intuitively, we would like $U_i$ to implement a Bloch sphere rotation by $\gamma$ for qubits $i$ with $v_i=1$ and by $0$ otherwise, up to equivalence modulo $\pi$. $V$ and $V^\dag$ effectively rotate the rotational axis of $U_i$ to point along the Bloch vector $\vec r$.

$U_i$ applies a rotation of magnitude $\sqrt{a^2+(cv_i)^2}t$, so we can construct the desired $U_i$ by solving for $c$ and $t$ such that
\begin{subequations}\label{eq:pipigamma}
\begin{align}
\label{eq:pipigamma1}
at &= \pi \hspace{.5cm} &v_i=0 \\ 
\label{eq:pipigamma2}
 \sqrt{a^2+c^2}t &= \pi+\gamma \hspace{.5cm} &v_i=1
\end{align}
\end{subequations}
The offset by $\pi$ ensures the system of equations is solvable for nonzero $t$ and real $c$. As it can be factored out as 
\begin{align}
    e^{-i (\pi+\gamma) \vec r \cdot \vec \sigma} = e^{-i \pi \vec r \cdot \vec \sigma} e^{-i \gamma \vec r \cdot \vec \sigma} = -e^{-i \gamma \vec r \cdot \vec \sigma}
\end{align} for unit vector $\vec r$, this offset's only effect on the dynamics is an overall phase. This system of equations is solved by $t=\pi/a$, $c = (a/\pi)\sqrt{(\pi+\gamma)^2-\pi^2}$ (recall that in the V$Z$ model $a$ is given and fixed). For qubits with $v_i=0$, this choice amounts to $U_i = e^{-i \pi X_i}$ and thus no net rotation. For qubits with $v_i=1$ the resulting action of $U_i$ is 

\bes
\begin{align}\label{eq:mid_rot}  
    U &= e^{-it(a X + c Z)} \\
    &= e^{-i(\pi+\gamma) (\sin(2\alpha)X+\cos(2\alpha)Z)}   
\end{align}
\ees
where $\alpha = \frac{1}{2}\cos^{-1}(\frac{c}{\sqrt{a^2+c^2}})$ [Fig.~\ref{fig:rot2_a}]. Thus the axis of rotation makes an angle $2\alpha$ with the $Z$-axis. 

\begin{figure*}
\subfigure[\ ]{\includegraphics[scale=.4]{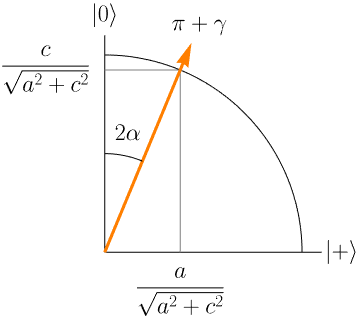}\label{fig:rot2_a}}
\subfigure[\ ]{\includegraphics[scale=.4]{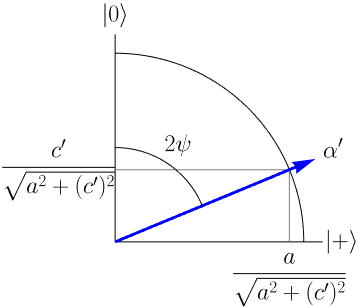}\label{fig:rot2_b}}
\subfigure[\ ]{\includegraphics[scale=.4]{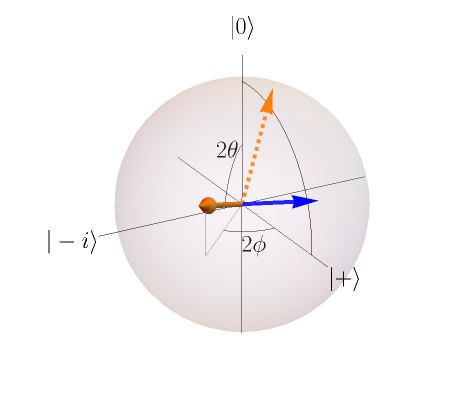}\label{fig:rot2_c}}
    \caption{(a) $U_i$ depicted by its rotational axis in the $X$-$Z$ plane of the Bloch sphere. Here $c$ is picked to make the magnitude of rotation $\pi+\gamma$ for qubits with $v_i=1$, which results in the rotational axis making an angle $2\alpha$ with the $\ket{0}$ state. (b) The unitary $V$ has $t'$ and $c'$ such that the magnitude of rotation is $\alpha'$ and the axis makes an angle $2\psi$ with the $|0\rangle$ state for all qubits. (c) The net result of all three rotations $VU_iV^\dag$ is that the rotational axis of $U_i$ (dashed orange) is rotated by $V$ (blue) to point along $\vec r$ (solid orange) with spherical coordinates ($2\theta$,$2\phi$). The net rotation is of magnitude $\pi+\gamma$ about the axis $\vec r$.}
    \label{fig:rot}
\end{figure*}

When solving for $V$ we can ignore the $v_i=0$ case, as $VV^\dag=I$. We choose $c'$, $t'$ in Eq.~\eqref{eq:V} for $V$ to implement a rotation [Fig.~\ref{fig:rot2_b}]:
\bes
\begin{align}\label{eq:side_rot}  
    V &= e^{-it'(aX+c'Z)}\\
    &= e^{-i\alpha'(\sin(2\psi) X+\cos(2\psi)Z)}  
\end{align}
\ees
in order to effectively rotate the $\sin(2\alpha)X+\cos(2\alpha)Z$ axis into the desired axis $\vec r$ [Fig.~\ref{fig:rot2_c}]. The necessary $\alpha'$ and $\psi$ are solved for in Appendix \ref{app:sol}. They are:
\begin{subequations}
\label{eq:psialpha'}
\begin{align}
\label{eq:psialpha'-1}
    &\psi = \frac{1}{2}\tan^{-1}\left(\frac{\cos(2\alpha)-\cos(2\theta)}{\sin(2\theta)\cos(2\phi)-\sin(2\alpha)}\right)\\
\label{eq:psialpha'-2}
    &\alpha' = \frac{1}{2}\sin^{-1}\left(\frac{\sin(2\theta)\sin(2\phi)}{\sin(2\psi - 2\alpha)}\right) ,
\end{align}
\end{subequations}
where we take $\psi \in [0,\frac{\pi}{2}]$ and $\alpha' \in [\frac{\pi}{4},\frac{3\pi}{4}]$. From here it is straightforward to solve for $c'$, $t'$.

The resulting values of $c$, $t$, $c'$, $t'$, in Eq.~\eqref{eq:UiV} are
\begin{subequations}
\begin{align}
\label{eq:solt}
    t &= \frac{\pi}{a}\\
    \label{eq:solc}
    c &= \frac{a}{\pi}\sqrt{(\pi+\gamma)^2-\pi^2}\\
    \label{eq:solt'}
    t' &= \frac{\alpha'\sin(2\psi)}{a}\\
    \label{eq:solc'}
    c'&= a \cot(2\psi)
\end{align}
\end{subequations}
for $\psi$ and $\alpha'$ in Eq.~\eqref{eq:psialpha'}. Substituting these values into Eq.~\eqref{eq:VUV} yields the desired $G_L$.
\end{proof}

We note that in the case that $\phi = 0$, Eq.~\eqref{eq:psialpha'} reduces to $\psi=(\theta+\alpha)/2$ and $\alpha' = \pi/2$. This case creates the $T$ gate when ($\theta$,$\gamma$) = ($0$,$\frac{\pi}{8}$) and Hadamard gate when ($\theta$,$\gamma$) = ($\frac{\pi}{8}$,$\frac{\pi}{2}$) so it is in fact sufficient to construct a universal set of single-qubit gates. The benefit of having found a way to represent general single-qubit unitaries is theoretical completeness and the potential simplicity of other UGSs.

\subsection{$ZZ$ coupling layers}
\label{subsection:3.2}

Now consider an effective two-qubit coupling layer $G_L$ corresponding to $ZZ$-coupling gates $g_{ij}= e^{-iC Z_i Z_j}$ with an arbitrary real coupling constant $C>0$, and acting on some but not all pairs of qubits connected by edges in $\mathcal{E}$. Such a layer decomposes as
\bes
\label{eq:logicalcouplayer}
\begin{align}
    G_L &= \prod_{(i,j)\in \mathcal{E}} (g_{ij})^{w_{ij}} \hspace{.25in} w_{ij}\in \{0,1\}\\
    &= e^{-iC\sum_{(i,j)\in \mathcal{E}}w_{ij}Z_i Z_j}  ,
\end{align}
\ees
where $w_{ij} = w_{ji} =1$ iff qubit $i$ couples to qubit $j$ and $0$ otherwise. Again the question arises, how would we implement $G_L$ using just the components of the V$Z$ model?

We restrict to the case where each qubit experiences at most a single two-qubit gate (and no $>2$-qubit gates) in a single timestep (we relax this restriction in Sec.~\ref{sec:depth-req}). I.e., we assume that each qubit couples to at most one of its neighbors at a time: 
\bes
\label{eq:w-cond}
\begin{align}
 \sum_j w_{ij} &= \sum_i w_{ij} \in \{0,1\}\\ 
 w_{ij} &= 0 \text{ for }(i,j) \notin \mathcal{E} 
\end{align}
\ees
In this case we have the following lemma:

\begin{mylemma}
\label{lem:couple}
The V$Z$ model can implement a two-qubit coupling gate layer $G_L$ [Eq.~\eqref{eq:logicalcouplayer}] in at most 6 applied layers.
\end{mylemma}

\begin{proof}
We show that this may be implemented in the V$Z$ model using either three or six applied layers depending on whether there exist uncoupled qubits. Let $\mathcal{S}\subset \mathcal{V}$ be the set of uncoupled qubits, i.e., qubits $i$ for which $w_{ij}=0\ \forall j$, and set $X_\mathcal{S} = \sum_{i \in \mathcal{S}} X_i$. The implementation takes the form:
\begin{equation}
\label{eq:couplayer}
G_L = e^{-i \gamma X_\mathcal{S}} e^{-i t'H^x} U e^{-it'H^x}
\end{equation}
for 
\begin{equation}
    U = e^{-it(H^x + b \sum_{(i,j)\in \mathcal{E}} w_{ij} Z_i Z_j)}
\end{equation}
and $e^{-i \gamma X_\mathcal{S}}$ an auxiliary single qubit unitary layer, as defined in the previous subsection, acting on qubits in $\mathcal{S}$. This $e^{-i \gamma X_\mathcal{S}}$ cancels out the $X$-rotation uncoupled qubits experience while coupled qubits are being acted upon; in the case that all qubits are coupled it reduces to the identity. The decomposition into applied layers is depicted in Fig.~\ref{fig:couplegatelayer}.

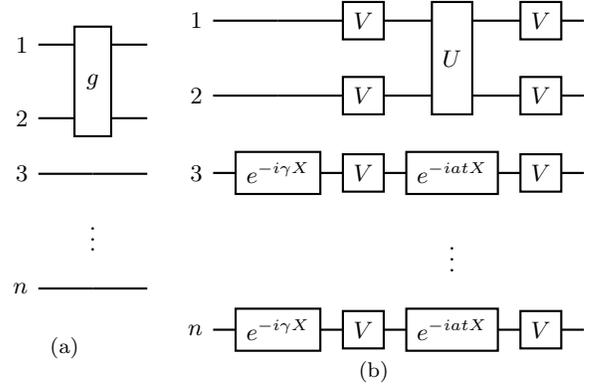
\begin{figure}
\subfigure[\ ]{\begin{quantikz}
    \lstick{1} &\gate[wires=2]{g}&\qw\\
     \lstick{2} &&\qw\\
    \lstick{3} &\qw&\qw\\
    &\vdots\\
    \lstick{$n$}&\qw&\qw\\
\end{quantikz}}\hspace{2mm}
\subfigure[\ ]{\begin{quantikz}[column sep=3mm]
     \lstick{1} &\qw &\gate{V} &\gate[wires=2]{U} &\gate{V} &\qw\\
     \lstick{2} &\qw &\gate{V} & &\gate{V} &\qw\\
    \lstick{3} &\gate{e^{-i\gamma X}} &\gate{V} &\gate{e^{-iatX}} &\gate{V} &\qw\\
    &&&\vdots\\
     \lstick{$n$} &\gate{e^{-i\gamma X}} &\gate{V} &\gate{e^{-iatX}} &\gate{V} &\qw
\end{quantikz}}
\caption{a) Example effective two-qubit gate layer $G_L$ with a gate $g$ acting on only the first 2 qubits. b) Its implementation as applied layers in the V$Z$ model. Here $V = e^{-iat'X}$.}
\label{fig:couplegatelayer}
\end{figure}

The pairwise coupling assumption [Eq.~\eqref{eq:w-cond}] allows use of the following two-qubit decomposition between qubits $i$ and $j$, derived in Appendix \ref{app:decomp}:
\bes
\label{eq:2qubitdecomp}
\begin{align}
\label{eq:2qubitdecomp-a}  
    &e^{-it(a (X_i+X_j) + b Z_i Z_j)}\\
    \label{eq:2qubitdecomp-b}  
    &= e^{-i\beta (X_i+X_j)}e^{-iD_1X_i X_j} e^{-iD_2Y_i Y_j}e^{-iD_3 Z_i Z_j}e^{-i\beta (X_i+X_j)} 
\end{align}
\ees
with
\begin{subequations}
\label{eq:2qubitsolns}
\begin{align}
\label{eq:solD1}
   &D_1 = 0 \\
   \label{eq:solD2}
   &D_2 = \frac{1}{2}(b t-\omega)\\
   \label{eq:solD3}
    &D_3 = \frac{1}{2}(b t+\omega)\\
    \label{eq:solomega}
    &\omega = \sin^{-1}\left(\frac{b}{\sqrt{4a^2+b^2}}\sin(t\sqrt{4a^2+b^2})\right)\\
    \label{eq:solbeta}
    &\beta = \frac{s}{4} \cos^{-1}\left(\cos(t\sqrt{4a^2+b^2})\sec(\omega)\right)+\frac{\pi}{2}    
\end{align}
\end{subequations}
where $s=\text{sign}(a \sin(\sqrt{4a^2+b^2}t))$ and we pick $\omega \in [-\frac{\pi}{2},\frac{\pi}{2}]$ and $\beta \in [0,\pi]$. Uncoupled qubits have $w_{ij}=0\ \forall j$, and simply experience $e^{-iatX}$, as shown in Fig.~\ref{fig:couplegatelayer}.

We restrict to the pure $ZZ$-coupling of Eq.~\eqref{eq:logicalcouplayer} between qubits $i$ and $j$ by requiring that  $D_2=0$ [to cancel the undesired $Y_iY_j$ term in Eq.~\eqref{eq:2qubitdecomp-b}], and thus that 
\beq
\label{eq:cond1}
b t = \omega = D_3 \equiv D .
\eeq 
For coupled pairs, we take $b \neq 0$ and solve for $t$ in terms of $D$. Namely, substituting Eq.~\eqref{eq:cond1} into Eq.~\eqref{eq:solomega} we obtain 
\beq
    \label{eq:sincs-4}
    \text{sinc}(D) = \text{sinc}\left(\sqrt{4a^2 t^2+D^2}\right) ,
\eeq
where $\text{sinc}(x) \equiv \sin(x)/x$.
 
Achieving the desired magnitude of coupling in Eq.~\eqref{eq:logicalcouplayer} up to an overall phase requires that $D = C$ mod $\pi$. As demonstrated in Appendix \ref{app:numeric}, for every $C \in [0,\pi]$, there exists a $k \in \{0,1,2,3\}$ such that when $D = C+k\pi$, Eq.~\eqref{eq:sincs-4} is numerically solvable for $at>0$. Thus for any $C \in [0,\pi]$, we can pick this value of $k$ and corresponding numerical solution $t$, and set $b = D t^{-1} = (C+k\pi) t^{-1}$. With these parameter choices and taking the product across all pairs of qubits, the overall action of $U$ becomes
\bes
\begin{align}
    &e^{-it(H^x + b \sum_{ij }w_{ij} Z_i Z_j)}\\
    &= e^{-i at X_\mathcal{S}} e^{-i \beta X_{\overline{\mathcal{S}}}} e^{-i C \sum_{ij} w_{ij} Z_i Z_j}e^{-i \beta X_{\overline{\mathcal{S}}}} ,
 \end{align}
 \ees
up to an overall phase, for $X_{\overline{\mathcal{S}}}=\sum_{i \notin \mathcal{S}} X_i$.

The effect of $\beta$ can be undone by a uniform $X$-rotation for time $at' = \pi-\beta$ before and after the coupling layer. This gives
\bes
\begin{align}
    &e^{-it' H^x}e^{-it( H^x+b \sum_{ij} w_{ij} Z_i Z_j)}e^{-i t'H^x}\\
    &= e^{-i (at-2\beta) X_\mathcal{S}} e^{-iCw_{ij} Z_i Z_j} .
\end{align}
\ees

In the case that $|\mathcal{S}|=0$, $X_\mathcal{S}=0$, and the result is pure $ZZ$ coupling of the desired magnitude $C$. In the case that some qubits are uncoupled, the extra rotation on those qubits can be undone with an auxiliary inhomogeneous $X$-rotation by angle $\gamma = 2\beta - at$. In this case we have 
\bes
    \begin{align}
    &e^{-i\gamma X_\mathcal{S}}e^{-it' H^x}e^{-it(H^x+b\sum_{ij}w_{ij} Z_i Z_j)}e^{-it'H^x}\\
    &= e^{-iC\sum_{ij} w_{ij} Z_i Z_j} .
    \end{align}
    \ees

In either case, by picking parameters
\begin{subequations}
    \begin{align}
    &t\ \text{s.t.}\ \ \text{sinc}(C+k\pi)=\text{sinc}(\sqrt{4a^2t^2+(C+k\pi)^2})\\
    &b = (C+k\pi)t^{-1}\\
    &t' = (\pi-\beta)a^{-1}\\
    &\gamma = 2\beta-at ,
    \end{align}
\end{subequations}
 the effective coupling layer is implemented in at most 6 applied layers.
\end{proof}

Alternatively, if the coupling layer is preceded or succeeded by single qubit unitary layers acting nontrivially on the exact same set of qubits, the required $e^{-it'H^x}$ evolution may be absorbed into the single qubit unitary layers, and thus become effectively free, as long as the auxiliary $e^{-i \gamma X_\mathcal{S}}$ has $\gamma$ tuned to compensate. Furthermore, if multiple effective $ZZ$-coupling layers share a set $\mathcal{S}$ of uncoupled qubits, their auxiliary $X$-rotations will commute with all the $ZZ$-couplings, and thus can be combined into a single auxiliary unitary. Both of these methods may be used to implement SWAP gates more efficiently, as is done in Section~\ref{sec:supremacy}.

It is also worth mentioning that an effective layer of disjoint  $e^{-i\frac{\pi}{8}Z\otimes Z}$ gates may be implemented in at most six applied layers, and one of disjoint CZ gates $e^{-i\frac{\pi}{4}Z_1 Z_2}e^{-i\frac{\pi}{4}(Z_1+Z_2)}$ in at most eight by absorbing one of the $e^{-it'H^x}$ rotations into the single qubit effective layer. Either case is sufficient for universality in combination with arbitrary single qubit gates. Though universality can be achieved using just $C=\frac{\pi}{4}$, the fact that we demonstrated the ability under the V$Z$ model to apply arbitrary (real) couplings is for more than just theoretical completeness, as it can reduce depth requirements.

\subsection{Depth Requirements}
\label{sec:depth-req}

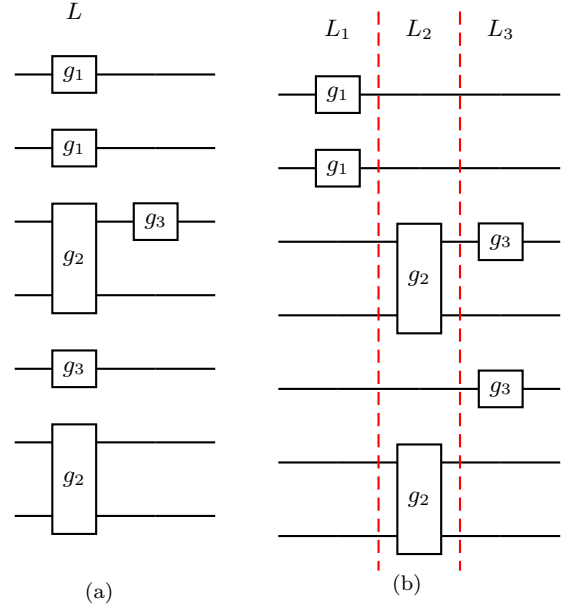
\begin{figure}
\subfigure[\ ]{\begin{quantikz}
&L&\\
&\gate{g_1}&\qw&\qw\\
&\gate{g_1}&\qw&\qw\\
&\gate[wires=2]{g_2}&\gate{g_3}&\qw\\
&&\qw&\qw\\
&\gate{g_3}&\qw&\qw\\
&\gate[wires=2]{g_2}&\qw&\qw\\
&&\qw&\qw\\
\end{quantikz}}
\hspace{5mm}%
\subfigure[\ ]{\begin{quantikz}
&L_1 & L_2 & L_3\\
&\gate{g_1} \slice{} & \qw \slice{} & \qw & \qw\\
&\gate{g_1} & \qw & \qw & \qw\\
&\qw & \gate[wires=2]{g_2} & \gate{g_3} & \qw\\
&\qw & & \qw & \qw\\
&\qw & \qw & \gate{g_3} & \qw\\
&\qw & \gate[wires=2]{g_2} & \qw & \qw\\
&\qw & & \qw & \qw
\end{quantikz}}
    \caption{(a) An example of a single effective layer $L$ of an arbitrary 1D circuit, in the case that $g_2$ and $g_3$ commute. The circuit is composed of gates $g_l$ of the UGS $\mathcal G$, which act on at most two neighboring qubits. (b) The layer $L$ implemented as three effective sublayers. Each sublayer consists of only a single gate and the identity, allowing it to be implemented in the V$Z$ model in constant depth using the methods of this section. 
    }
    \label{fig:depths}
\end{figure}

Naively, a depth $d$ quantum circuit on some graph could be implemented in the V$Z$ model on the same graph in $O(nd)$ layers without parallelization, simply by viewing each gate as its own effective layer. While this is sufficient for universality, we can reduce the required depth by a factor of $n$ using parallelization in terms of a universal gate set, leading to Theorem~\ref{thm:1}.

\begin{proof}[Proof of Theorem \ref{thm:1}]
Say that a circuit has depth $d$ when written in terms of UGS $\mathcal G$, composed of gates that act only on nearest neighbor qubits in some underlying architecture described by a graph $(\mcV,\mcE)$. We can assume that $\mathcal G$ is composed of single-qubit unitaries and two-qubit $ZZ$-coupling, as any two-qubit gate can be decomposed into $O(1)$ gates of this form~\cite{Barenco1995}. Then each effective layer of the original circuit can be viewed as at most $|\mathcal G|$ individual effective sublayers, each corresponding to the application of a single $g \in \mathcal G$ to some subset of the qubits (Fig.~\ref{fig:depths}). 

Consider an effective sublayer corresponding to a two-qubit coupling gate $g$. Within the sublayer, $g$ can act on an individual qubit $i$ once for each neighbor that qubit has. Because the decomposition used in Sec.~\ref{subsection:3.2} requires pairwise coupling of qubits, we must further decompose $g$ into an effective sublayer for each application of $g$ on $i$ (Fig.~\ref{fig:pairwise}). 
The application of a single two-qubit coupling gate layer may require as many effective sublayers as the minimum edge coloring of the underlying graph $(\mcV,\mcE)$. By Vizing's theorem~\cite{vizing}, this can be upper bounded by $\Delta+1$, where $\Delta$ is the maximum degree of the graph.
In case the graph is a grid based lattice of dimension $D$ , this simply becomes the number of sublattices of paired qubits, which is $\Delta=2D$. 

After full sublayer decomposition, each effective sublayer corresponds either to a single-qubit unitary or $ZZ$-coupling on monogamously paired qubits, so it is of a form that can be implemented in the V$Z$ model in $O(1)$ applied layers. An arbitrary quantum circuit of depth $d$ can then be represented in the V$Z$ model in depth $O(d|\mathcal G|\Delta)$ on the same graph. In the case of a $D$-dimensional integer lattice $\mathbb{Z}^D$, this becomes $O(2d|\mathcal{G}|D)$.
\end{proof}

For example to reproduce the effect of $\mathcal G = \{W, T, e^{-i\frac{\pi}{8} Z \otimes Z}\}$ on a 1D chain of qubits, effective sublayers of $W$ and $T$ gates may each be implemented in three applied layers while $e^{-i\frac{\pi}{8} Z \otimes Z}$ requires up to 2 effective sublayers of at most 6 each, so each effective layer of the original circuit requires at most 18 applied layers of the V$Z$ model to implement on the same graph.

By extending the 1D implementation of arbitrary IQP circuits used in Ref.~\cite{bermejo} to more general circuits, one sees that 1D quantum circuits are universal with $O(n)$ depth overhead. Thus the 1D V$Z$ model is capable of representing any depth-$d$ quantum circuit expressed in terms of UGS $\mathcal G$ on an arbitrary graph, in depth $O(n d |\mathcal G|$). Interestingly, as a fully connected graph has $\Delta = n-1$, in the worst case the 1D and same graph V$Z$ models have the same asymptotic depth scaling.

Note that in Table~\ref{tab:1} we state that an $O(d)$-depth 1D V$Z$ model reproduces the output distribution of a 1D circuit of $n$ logical qubits and depth $d$; the extra $n$ prefactor overhead in the expression $O(n d |\mathcal G|$) is incurred when going from an arbitrary graph to a 1D chain, or when generating SWAP gates.

The 1D case is also interesting in its own right: it is known that 1D $O(n)$ depth universal quantum circuits anticoncentrate~\cite{hangleiter} and can generate distributions that cannot be efficiently classically simulated~\cite{bermejo}. In the next section, we provide a circuit distribution that demonstrates this, and analyze the depth overhead required to implement it in the V$Z$ model.

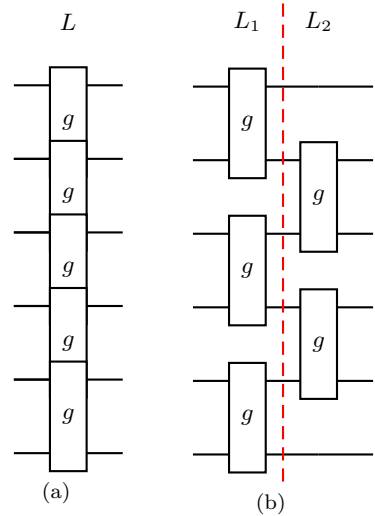
\begin{figure}
\subfigure[\ ]{\begin{quantikz}
&L &\\
&\gate[wires=2]{g} &\qw\\
&\gate[wires=2]{g} &\qw\\
&\gate[wires=2]{g} &\qw\\
&\gate[wires=2]{g} &\qw\\
&\gate[wires=2]{g} &\qw\\
& &\qw
\end{quantikz}}
\hspace{5mm}%
\subfigure[\ ]{\begin{quantikz}
&L_1\slice{} &L_2 &\\
&\gate[wires=2]{g} &\qw &\qw\\
& &\gate[wires=2]{g} &\qw\\
&\gate[wires=2]{g} & &\qw\\
& &\gate[wires=2]{g} &\qw\\
&\gate[wires=2]{g} & &\qw\\
& &\qw &\qw
\end{quantikz}}
    \caption{(a) An example of an effective coupling sublayer $L$ of a 1D circuit, in which the coupling gate $g$ is applied across all qubits, rather than restricted to pairwise coupling. (b) The sublayer $L$ implemented as two effective coupling sublayers, each with pairwise coupling.  
    }
\label{fig:pairwise}
\end{figure}

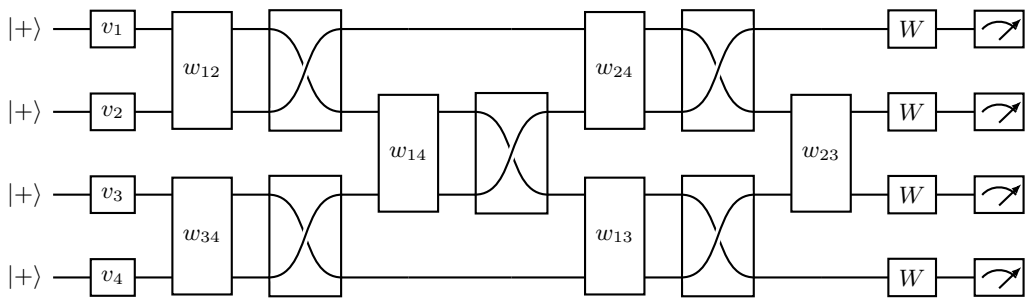
\begin{figure*}
\begin{quantikz}
    \lstick{$\ket{+}$} & \gate{v_1} & \gate[wires=2]{w_{12}} & \gate[swap]{} & \qw & \qw & \gate[wires=2]{w_{24}} & \gate[swap]{} & \qw & \gate{W} & \meter{}\\
    \lstick{$\ket{+}$} & \gate{v_2} & & & \gate[wires=2]{w_{14}} & \gate[swap]{} & & & \gate[wires=2]{w_{23}} & \gate{W} & \meter{}\\
    \lstick{$\ket{+}$} & \gate{v_3} & \gate[wires=2]{w_{34}} & \gate[swap]{} & & & \gate[wires=2]{w_{13}} & \gate[swap]{} & & \gate{W} & \meter{}\\
    \lstick{$\ket{+}$} & \gate{v_4} & & & \qw & \qw & & & \qw & \gate{W} & \meter{}
\end{quantikz}
    \caption{The $n=4$ case of a 1D $O(n)$-depth circuit capable of generating the distribution in Eq.~\eqref{eq:IQP} via alternating layers of $ZZ$ coupling and SWAP gates. It is equivalent to that of Ref.~\cite{bermejo}, but composed of Hadamard gates ($W$), SWAP gates, and $Z$-diagonal gates of the form $e^{i v_i Z_i}$ and $e^{i w_{ij} Z_i Z_j}$, denoted by their degree of rotation $v_i$ or $w_{ij}$. In this design qubits effectively act as particles moving past each other in 1D, and after $n$ layers each qubit has a chance to interact with each other qubit once. As $Z$-diagonal gates commute, multiple interactions between any pair of qubits can be implemented as a single interaction. Thus, this 1D depth $O(n)$ circuit is sufficient to implement the circuit implied by
    Eq.~\eqref{eq:IQP}.}
    \label{fig:circuit}
\end{figure*}

\begin{figure*}
\subfigure[\ ]{\begin{quantikz}
    & \gate{\frac{3\pi}{4}} \slice{} & \gate[wires=2]{\frac{5\pi}{8}} & \qw\\
    & \gate{\frac{5\pi}{8}} & & \qw\\
    & \gate{\frac{\pi}{4}} & \gate[wires=2]{\frac{\pi}{4}} & \qw\\
    & \gate{\frac{\pi}{8}} & & \qw
\end{quantikz}\label{fig:5a}}
\hspace{5mm}
\subfigure[\ ]{\begin{quantikz}
    & \gate{\frac{\pi}{2}} & \gate{\frac{\pi}{4}} & \qw \slice{} & \gate[wires=2]{\frac{\pi}{2}} & \qw & \gate[wires=2]{\frac{\pi}{8}} & \qw\\
    & \gate{\frac{\pi}{2}} & \qw & \gate{\frac{\pi}{8}} \slice{} & & \qw & \qw & \qw\\
    & \qw & \gate{\frac{\pi}{4}} & \qw \slice{} & \qw & \gate[wires=2]{\frac{\pi}{4}} & \qw\\
    & \qw & \qw & \gate{\frac{\pi}{8}} \slice{} & \qw & & \qw
\end{quantikz}\label{fig:5b}}
    \caption{(a) The first two layers of an example circuit with arbitrarily chosen $w_{ij}$ and $v_i$. 
    $Z$-diagonal gates are denoted by their degree of rotation.
    (b) The same two layers rewritten as effective sublayers which can be implemented in the V$Z$ model. 
    For example, the gate corresponding to $v_2=\frac{5\pi}{8}$ in (a) is decomposed into $\frac{\pi}{2}$ and $\frac{\pi}{8}$ gates in (b), with the $\frac{\pi}{8}$ gate shifted so that it is executed in parallel with the other $\frac{\pi}{8}$ gate from (a) to form a sublayer.}
    \label{fig:bermejo1}
\end{figure*}
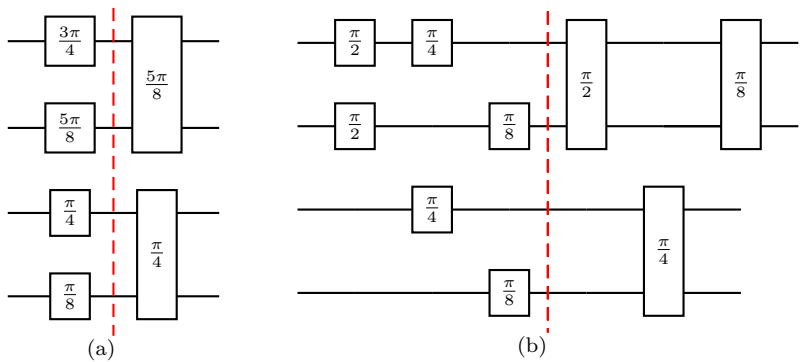

\section{Demonstrating quantum supremacy}
\label{sec:supremacy}

\subsection{Homogeneous $X$-field}
In this section we apply the previous results towards demonstration of quantum supremacy.
\begin{mycorollary}
The 1D V$Z$ model can generate a distribution in $O(n)$ depth that cannot be approximately simulated by a classical computer in the worst case.
\end{mycorollary}

The proof follows from Theorem \ref{thm:1} and the distribution generated in Refs.~\cite{bremner1, bermejo}, in which it is shown that no classical computer can efficiently sample from a distribution approximating the one in question. Here we detail the distribution and account for the exact overhead introduced from generating it in the V$Z$, rather than gate-based model of quantum computation. This distribution was proposed in Ref.~\cite{bremner1} and is the output distribution of an Instantaneous Quantum Polynomial-time (IQP) model with long range interactions:
\begin{equation}
\label{eq:IQP}
    P_{\text{IQP}}(\textbf{s}) = \langle \mathbf{s} | H^{\otimes n} e^{-iC_z} \ket{\bm{+}} ,
\end{equation}
where
\begin{align}
& C_z = w_{ij}Z_i Z_j + v_i Z_i \ , \ \  w_{ij}, v_i \in \left\{\frac{1}{8}k \pi\right\}_{k=0}^7 ,
\end{align}
with $w_{ij}$ and $v_i$ chosen uniformly at random from their domains and \textbf{s} a string in the $Z$ basis. 

In the V$Z$ model we can generate the distribution of Eq.~\eqref{eq:IQP} using a modified version of the circuit design of Ref.~\cite{bermejo}, as is depicted in Fig.~\ref{fig:circuit}. The depth overhead of implementing this circuit in the V$Z$ model is affected by choice of UGS. Rather than implementing each term in $C_Z$ as an independent single qubit or coupling gate [Fig.~\ref{fig:5a}], we break each layer of single qubit and coupling gates into three effective sublayers by using the binary decomposition
\begin{equation}
\label{eq:binary}
    v_i = \frac{4\pi}{8}a_i+\frac{2\pi}{8}b_i+\frac{\pi}{8}c_i \hspace{.5cm} a_i,b_i,c_i \in \{0,1\}
\end{equation}
and implementing each $a$, $b$, $c$ as its own sublayer. Each $ZZ$ coupling layer is decomposed into three sublayers with an equivalent decomposition on $w_{ij}$ [Fig.~\ref{fig:5b}]. A SWAP gate may be implemented by decomposition into a product of Hadamard gates and single-qubit and two-qubit $Z$-diagonal gates (Fig.~\ref{fig:6}).

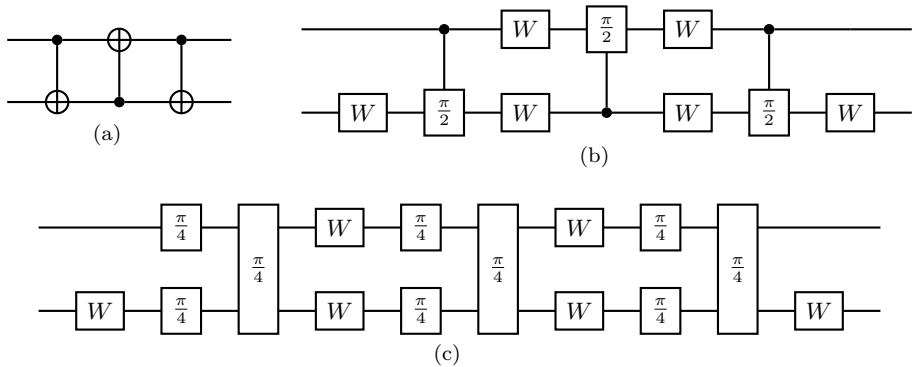
\begin{figure*}
\subfigure[\ ]{\begin{quantikz}
& \ctrl{1} & \targ{} & \ctrl{1} & \qw\\
& \targ{} & \ctrl{-1} & \targ{} & \qw
\end{quantikz}\label{fig:6a}}
\hspace{5mm}%
\subfigure[\ ]{\begin{quantikz}
& \qw & \ctrl{1} & \gate{W} & \gate{\frac{\pi}{2}} & \gate{W} & \ctrl{1} & \qw & \qw\\
& \gate{W} & \gate{\frac{\pi}{2}} & \gate{W} & \ctrl{-1} & \gate{W} & \gate{\frac{\pi}{2}} & \gate{W} & \qw
\end{quantikz}\label{fig:6b}}
\subfigure[\ ]{\begin{quantikz}
& \qw & \gate{\frac{\pi}{4}} & \gate[wires=2]{\frac{\pi}{4}} & \gate{W} & \gate{\frac{\pi}{4}} & \gate[wires=2]{\frac{\pi}{4}} & \gate{W} & \gate{\frac{\pi}{4}} & \gate[wires=2]{\frac{\pi}{4}} & \qw & \qw\\
& \gate{W} & \gate{\frac{\pi}{4}} & & \gate{W} & \gate{\frac{\pi}{4}} & & \gate{W} & \gate{\frac{\pi}{4}} & & \gate{W} & \qw
\end{quantikz}\label{fig:6c}}
    \caption{The SWAP gate represented as (a) a product of CNOT gates, (b) Hadamard ($W$) and CZ gates, and (c) Hadamard, $e^{-i\frac{\pi}{4}Z}$, and $e^{-i\frac{\pi}{4}Z \otimes Z}$ rotation gates. The latter may be implemented in the V$Z$ model using the methods of Sec.~\ref{sec:model}.}
    \label{fig:6}
\end{figure*}

Using the methods of Sec.~\ref{sec:model}, each sublayer of Fig.~\ref{fig:5b} corresponding to a single qubit gate can be implemented in three applied layers, and thus the set of all single qubit gates requires nine applied layers using the decomposition of Eq.~\eqref{eq:binary}. Similarly, each coupling layer requires $18$ applied layers, or $18n$ in total between the $n$ $ZZ$-coupling layers. In implementation of the SWAP layer [Fig.~\ref{fig:6c}], we combine adjacent single qubit unitary layers acting nontrivially on the same qubits. The result is that a SWAP gate requires only five effective single qubit unitary layers, or 15 applied layers. We further combine the $e^{-it'H^x}$ rotations needed for coupling layers into the adjacent single qubit unitary layers; all but one $e^{-it'H^x}$ rotation can be absorbed this way. All the auxiliary $e^{-i\gamma X_\mathcal{S}}$ rotations can be combined, so the set of three effective $ZZ$-coupling layers requires only four applied layers to implement if every qubit is coupled ($|\mathcal{S}|=0$), or seven otherwise. In total, each of the $n$ effective SWAP layers requires either $19$ or $22$ applied layers, depending on whether the SWAP gates span every qubit or not. When $n$ is even, half of the SWAP layers act on all qubits, but for odd $n$ every layer has uncoupled qubits, so in general the total overhead for the SWAP gate layers is at most $22n$. Summing the $18n$ applied layers for $ZZ$-coupling, $22n$ applied layers for SWAP gates, nine applied layers for single-qubit gates, and final all-qubit Hadamard, the entire circuit can be implemented in $40n+10$ applied layers. 

These calculations assume a noiseless computation model. The presence of noise can reduce fidelity of an output distribution of a circuit to the point that it no longer demonstrates quantum supremacy (see, e.g., Ref.~\cite{Zlokapa:2020} for a discussion of this in the context of random circuit sampling), but for this distribution even a simple repetition code with $O(\text{ln}(n))$ repetitions is sufficient to regain supremacy~\cite{bremner2}. The universality of the V$Z$ model allows it to perform this repetition code efficiently, in contrast to methods of demonstrating QS which are not universal, and thus not necessarily able to correct for the effects of decoherence.

 \subsection{Alternating $X$-field}
In this section we consider a slight modification of the V$Z$ model in which the $X$-field can be turned off for half the qubits at a time -- either those of even or odd index. This modification reduces the depth required to apply the SWAP gate while still being easier to implement than the full circuit model, and thus reduces total depth required to demonstrate quantum supremacy. We still use the homogeneous $X$-field method for implementing the $ZZ$-coupling layers. For even $n$, in each SWAP layer the SWAP gates alternate between acting pairwise on one sub-lattice such that they cover all of the qubits, and acting pairwise on the other sub-lattice such that there are two unaffected qubits at the boundaries of the qubit chain. For odd $n$ every sub-lattice acted on by SWAP gates has a single unaffected qubit. Cases with unaffected qubits require more layers to cancel the $X$-rotations these boundary qubits experience while the other qubits are being swapped, which we will discuss it at the end of this subsection.

In the alternating $X$-field method, instead of implementing the SWAP gate as three CNOT gates as in Fig.~\ref{fig:6a}, we use different building blocks: $ZZ$-coupling layers with nonzero $X$-fields on either only the even or only the odd qubits, as shown in Fig.~\ref{fig:1Dchain}. First we start from one of these coupling layers with $X$-fields on only the even qubits. We can write the action of this layer as
\begin{align}
U &= \prod_{i=\text{odd}} e^{-i t (a X_{i+1} + b Z_i Z_{i+1})} 
\end{align}
for appropriately chosen $b$. Consider just the two-qubit gate acting on the first and second qubits; it has the following block-diagonal structure:
\be
U_{12} = e^{-i t (a X_2 + b Z_1 Z_2)} = 
\begin{bmatrix} 
U_1& \bm{0}  \\
\bm{0} & U_2 \\
\end{bmatrix}
\ee
where $\bm{0}$ is the $2\times 2$ zero-matrix,
and we can further decompose 
\bes
    \begin{align}
    \label{eq:decomp-1}
        U_1 &= e^{i \alpha Y } e^{-i \gamma X} e^{- i \alpha Y} \\
        U_2 &= e^{-i \alpha Y } e^{-i \gamma X} e^{i \alpha Y},
    \end{align}
\ees
where
\begin{subequations}\label{eq:decomp}
\begin{align}
\label{eq:decomp-2}
\alpha &= \frac{1}{2} \cos^{-1} \left(\frac{a}{\sqrt{a^2 + b^2}}\right)\\
\label{eq:decomp-4}
\gamma &= t \ \sqrt{a^2 + b^2} .
\end{align}
\end{subequations}
If we set $b=a$ and $t = \frac{\pi}{2 \sqrt{2}a}$, then $U_1 = -i W$, and $U_1^\dag U_2 = \tilde{Y}$, where 
$\tilde{Y} = ZX = 
\begin{bmatrix} 
0 & 1  \\
-1 & 0 \\
\end{bmatrix}$. 
Thus for this choice of $b$ and $t$, 
\begin{equation}
 \label{eq:U12}
 U_{12} = -i(I\otimes W)C\tilde{Y}_{12}
\end{equation}
where $C\tilde{Y}_{12}$ is the controlled $\tilde{Y}$ gate. Similar to the $ZZ$-coupling gates (recall the discussion in Sec.~\ref{subsection:3.2}), this operation leads to a net $X$-rotation on uncoupled qubits which still experience an $X$-field, which then must be undone by an extra single-qubit unitary following the SWAP layer.

 Defining the above building block $U_{12}$, the total circuit required for a SWAP gate will be:
 \begin{subequations}
  \label{eq:U12prod}
\begin{align}
\label{eq:U12prod-1}
     C &= U_{12} W_1 W _2 U_{21} U_{12} W_2  \\
\label{eq:U12prod-2}
    &=i(W_2 C\tilde{Y}_{12}) W_1 W _2 (W_1 C\tilde{Y}_{21}) (W_2 C\tilde{Y}_{12}) W_2
\end{align}
 \end{subequations}
The product of $W_1 W_2$ needs just one effective single qubit gate layer in total, and any effect on uncoupled qubits can be incorporated into the extra single-qubit unitary $e^{-i\gamma X_\mathcal{S}}$ on these qubits following the SWAP layer. Since the last $W_2$ requires three single qubit gate layers, the total number of required layers for implementing the above decomposition is seven.
 
 Defining $\tilde{\bar{Y}} = - \tilde{Y} = XZ$ and using the equality $W_2 C\tilde{Y}_{12} = C\tilde{\bar{Y}}_{12} W_2$, we can simplify the above expression as below:
\begin{subequations} \label{eq:U12prodproof}
\begin{align}
    \label{eq:U12prodproof-1}
 C&= iW_2 C\tilde{Y}_{12} W_1 W_2 W_1 C\tilde{Y}_{21} W_2 C\tilde{Y}_{12} W_2  \\
 \label{eq:U12prodproof-2}
 &= iC\tilde{\bar{Y}}_{12} W_2 W_1 W_2 W_1 C\tilde{Y}_{21} C\tilde{\bar{Y}}_{12} W_2 W_2 \\
 \label{eq:U12prodproof-3}
 &= iC\tilde{\bar{Y}}_{12} C\tilde{Y}_{21} C\tilde{\bar{Y}}_{12} \\
 \label{eq:U12prodproof-4}
 &= i CX_{12} CZ_{12} CZ_{21} CX_{21} CX_{12} CZ_{12}  \\
 \label{eq:U12prodproof-5}
 &= i CX_{12}  CX_{21} CX_{12} CZ_{12} \\
 \label{eq:U12prodproof-6}
 &= i \text{SWAP}_{12} CZ_{12} \\
 \label{eq:U12prodproof-7}
 &= i CZ_{21} \text{SWAP}_{12} .
 \end{align}
 \end{subequations}
This is the required SWAP gate up to $CZ_{21}$. This unwanted $CZ_{21}$ can be written as a product of single qubit $Z$-diagonal gates and $ZZ$-couplings, and can in fact be incorporated into the existing single qubit and coupling gates already in the overall circuit, simply by shifting the values of $w_{12}$, $v_1$, and $v_2$ by $\frac{\pi}{4}$. These $v_i$ and $w_{ij}$ parameters are chosen from a uniform distribution over the full range of multiples of $\frac{\pi}{8}$ mod $\pi$, so the shift by $\frac{\pi}{4}$ does not affect the distribution. This makes the $CZ$ gate effectively free to include.

In this construction a SWAP gate layer spanning all qubits can be implemented in seven applied layers. In the case of unswapped qubits, an unswapped qubit experiences a different net $X$-rotation depending on whether it is of even or odd qubit index. As there are at most two values of unwanted $X$-rotation after the SWAP layer, these can be undone in at most two single-qubit effective layers, or six applied layers. This brings the total number of applied layers to implement a single effective SWAP layer to seven if every qubit is swapped, $10$ if there is a single unswapped qubit, or at most $13$ otherwise.

For even $n$ the model alternates between SWAP layers with zero and two unswapped qubits, so the set of all SWAP layers requires $7\frac{n}{2}+13\frac{n}{2}=10n$ applied layers. For odd $n$, every SWAP layer has a single unswapped qubit, so the SWAP layers require a total of $10n$ applied layers -- the same as in the even $n$ case. Including $18n$ applied layers of $ZZ$-coupling gates and $10$ of single qubit gates and Hadamards, the total number of applied layers required to reach supremacy becomes $28n+10$ (recall that the cost was $40n+10$ in the case of a fully homogeneous $X$-field). 
 
\begin{figure}[t]
 \includegraphics[scale=.75]{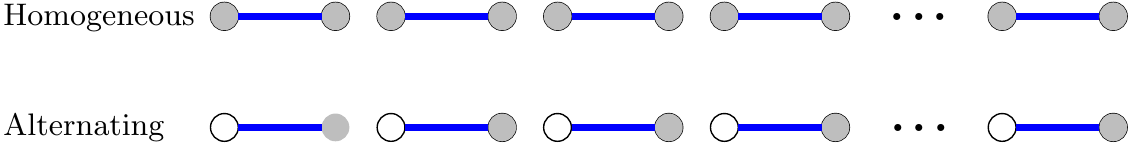}
 \caption{Homogeneous vs alternating pattern of $X$-fields. Blue lines represent $ZZ$-couplings, which form a sub-lattice. Grey circles represent nonzero $X$-field, while white circles have no $X$-field.}
 \label{fig:1Dchain}
 \end{figure}

\section{Conclusions}
\label{sec:conclusion}

Minimalism is both a sound engineering design principle and a desirable feature of theoretical models. By minimalism we mean a reduction in the consumption of some or more precious resources, or the reliance on a small set of assumptions. This work is an attempt to be minimalistic about the resources and assumptions underlying both implementations and models of universal quantum computation and quantum supremacy. From a practical perspective, spatially homogeneous or global control can be a significant advantage over individualized or local control, since the latter typically involves additional control wires or lasers, each of which is another source of noise and decoherence. From a theoretical perspective, it is interesting to investigate the computational power of fully or partially translationally invariant models in low dimensions.  

In accordance with this perspective we have studied here the power of a minimalistic model that assumes spatially homogeneous $X$-field control but spatially inhomogeneous $Z$-field control [Eq.~\eqref{eq:H_l}]. We have shown that this ``V$Z$ model" can be used to demonstrate quantum supremacy even in 1D, and is furthermore universal for quantum computation. The overhead required to implement the circuit model with a standard universal gate set using the 1D V$Z$ model is constant in system size and is significantly lower than the universal version of QAOA (see Table~\ref{tab:1}).
We are unaware of any model of quantum computation with a higher degree of homogeneity that achieves the same asymptotic depth scaling as the V$Z$ model.

Perhaps the most immediate application of the V$Z$ model would be a demonstration of quantum supremacy in the IQP setting, which we have shown here can be implemented in at most $40n+10$ layers of control pulses, for a circuit of width $n$. This model may be the next to be used for a quantum supremacy demonstration, now that this has been done using random circuit sampling~\cite{Arute:2019aa} and boson sampling~\cite{Zhong:2020aa}. Such a demonstration within the V$Z$ model would be particularly viable using flux qubits~\cite{harris_flux_qubit_2010,Yan:2016aa,grover2020fast}, where leaving a constant and homogeneous $X$-field on while locally controlling only $Z$-fields and interactions is both simpler and less prone to noise than also controlling the $X$-field.

\acknowledgments
We thank Paolo Zanardi for many insightful discussions. This research was sponsored by the Army Research Office and was accomplished under Grant Number W911NF-20-1-0075. This research is also based upon work (partially) supported by the Office
of the Director of National Intelligence (ODNI), Intelligence Advanced Research Projects Activity (IARPA) and the Defense
Advanced Research Projects Agency (DARPA), via the U.S. Army Research Office contract W911NF-17-C-0050. The views
and conclusions contained herein are those of the authors and should not be interpreted as necessarily representing the official
policies or endorsements, either expressed or implied, of the ODNI, IARPA, DARPA, ARO, or the U.S. Government. The
U.S. Government is authorized to reproduce and distribute reprints for Governmental purposes notwithstanding any copyright
annotation thereon.
This work is partially supported by a DOE/HEP QuantISED
program grant, QCCFP/Quantum Machine Learning and Quantum
Computation Frameworks (QCCFP-QMLQCF) for HEP, Grant No.
DE-SC0019219.

%%%%%%%%%%%%%%

\begin{appendix}

%\section{X field in Flux Qubits}\label{app:flux}
%A major source of decoherence in superconducting flux qubits is noise in qubit frequency, which is often caused by noise in qubit $Z$-bias lines. In the regime in which flux qubits are usually operated, the Hamiltonian is approximately
%\begin{equation}
%    H = A\sigma_x + B\sigma_z
%\end{equation}
%where $B$ is the $Z$-bias. The energy gap is $2\sqrt{A^2+B^2}$, which for non-negligible $A$ is insensitive to first order in bias line noise. However, for negligible $A$, the energy gap is $\sim 2|B|$, and bias line noise leads directly to noise in qubit frequency. 

\section{Complexity basis of quantum supremacy}
\label{app:complex}

Bounded-error Probabilistic Polynomial-time (BPP) is the computational complexity class of problems that a classical computer, with access to true randomness, can solve with high probability in time polynomial in input size. It can be viewed as the class of problems that is it realistic for a classical computer to solve. Bounded-error Quantum Polynomial-time (BQP) is the equivalent class for quantum computers. It is widely believed that BQP $\neq$ BPP, due to a hypothetical tool called postselection. Postselection is the ability to select the results of randomness after completing all computations. The classes associated with BPP and BQP when the computer has access to postselection are referred to as PostBPP and PostBQP. While PostBPP is known to be within the third level of the polynomial hierarchy, a hierarchy representing a generalization of the P vs NP distinction, it is known by Toda's and Aaronson's theorems that PostBQP in fact contains the entire polynomial hierarchy as a subset \cite{toda1991, aaronson2016}. If BQP=BPP, then postBQP=postBPP, and the entire polynomial hierarchy would be contained within its third level, a situation referred to as collapse of the polynomial hierarchy. This is conjectured to not happen for the same reason it is believed that P $\neq$ NP.

The output of problems in BQP takes the form of distributions that one samples from. Because we cannot take infinite samples in practice, we cannot perfectly calculate the distribution that a practical quantum computer is sampling from. Verifying that a problem has been solved relies on anticoncentration of the output distribution \cite{bouland2018, bremner2}: the outcome probabilities must be sufficiently spread across all possible outcomes of measurements, so that there is a nonvanishing signal-to-noise ratio for at least a constant fraction of all possible outcomes.

\section{QAOA vs the V$Z$ model}
\label{app:notQAOA}

In QAOA~\cite{farhi2014quantum} one considers a Hamiltonian of the form
\bes
    \label{eq:HQAOA}
\begin{align}  
    H(t) &= a(t) H^X + d(t) H^Z  \\
    \label{eq:Hx}    
    H^X &= \sum_{i\in\mathcal{V}} X_i \\ 
    H^Z  &=\sum_{(i,j)\in\mathcal{E}} b_{ij} Z_i Z_j + \sum_{i\in\mathcal{V}} c_i Z_i  ,
    \label{eq:Hz}
\end{align} 
\ees
for $n$ qubits occupying the vertices $\mathcal{V}$ of a graph $\mathcal{G}=\{\mathcal{V},\mathcal{E}\}$.
The parameters $b_{ij}$ and $c_{i}$ are controllable longitudinal local field and coupling constants, respectively, and $a(t)$ and $d(t)$ are time-dependent control functions. Note that $H_l$ [Eq.~\eqref{eq:H_l}] is an instance of $H(t)$, for each fixed $l$ and $t$.

The level-$p$ QAOA produces an approximation $C_p^*$ to the optimal value of the classical cost function represented by $H^Z$:
\bes
\begin{align}
\label{eq:QAOA-def}
U(\beta, \gamma) &= e^{-i \beta H^X} e^{-i \gamma H^Z}\\
|\psi(\boldsymbol{\gamma},\boldsymbol{\beta}) &= \left(\prod_{k =1}^p U(\beta_k , \gamma_k) \right) |\bm{+}\rangle\\
C_p^* = &\min_{\boldsymbol{\gamma},\boldsymbol{\beta}} \; \; \langle\psi(\boldsymbol{\gamma},\boldsymbol{\beta}) | H^Z | \psi(\boldsymbol{\gamma},\boldsymbol{\beta}) \rangle
\end{align}
\ees
where $\boldsymbol{\gamma} = (\gamma_1 , ... , \gamma_p) ,\boldsymbol{\beta} = (\beta_1,...,\beta_p)$ are the angles that parameterize the circuit.  Various heuristic methods for choosing these angles have been considered, and for small values of $p = O(1)$ the optimization can be done exactly~\cite{szegedy2019qaoa}.  

The V$Z$ model (Def.~\ref{def:TZ}) differs from 
QAOA in that the alternating sequence of unitaries always includes $H^X$, unlike the QAOA sequence given in Eq.~\eqref{eq:QAOA-def}. It also differs in that in QAOA the $H^Z$ Hamiltonian is fixed, while in the V$Z$ model we assume that the $b$,$c$,$w$, and $v$ coefficients in Eq.~\eqref{eq:H_l} are controllable from layer to layer, which would be equivalent to making $b$ and $c$ time-dependent in Eq.~\eqref{eq:Hz}.

\section{Depth Scaling of QAOA Universality}
\label{app:lloyd}
Here we analyze the required depth of the QAOA model used in Ref.~\cite{lloyd} to reproduce a 1D circuit to within a given total variation distance error $\epsilon$. This model of QAOA uses co-irrationality of terms in a $Z$-diagonal Hamiltonian $H_Z$ to reproduce a translationally invariant gate layer $U$ composed of a tensor product of one to two qubit gates across all qubits. The evolution time needed to reproduce a specific gate layer such that
\begin{equation}
    \|e^{-itH_Z}-U\| \leq \epsilon'
\end{equation}
is $t \in O(\epsilon'^{-4})$, where $\|\cdot \|$ is the operator norm. 

The BCQA method applies individually addressed gates by using SWAP gates to effectively walk a certain qubit, dubbed the control unit, across the chain until it is adjacent to a qubit being acted on~\cite{Simon}, then performing controlled gates. It thus requires $O(n)$ depth to implement a single gate. Assuming the ability to parallelize, a 1D circuit of depth $d$ with $O(nd)$ gates will then require $O(nd)$ depth BQCA to reproduce. In the QAOA implementation this becomes $O(nd\epsilon'^{-4})$, but we would like to bound runtime in terms of total variation distance $\epsilon$ of the output distribution, not operator norm error $\epsilon'$ of each individual gate. 

For a target BQCA protocol $C = \prod_{l=1}^{nd}U_l$, we can write the QAOA approximation as $C' = \prod_{l=1}^{nd}(U_l+\delta_l)$, where $\delta_l$ is the error of the $l^{\text{th}}$ gate layer, such that $\|\delta_l\| \leq \epsilon'\ \forall l$. Let $\Delta \equiv C'-C$. Then the operator norm error of the complete QAOA implementation becomes
\bes
    \begin{align}
    \|\Delta\| &= \|\prod_{l=1}^{nd}(U_l+\delta_l) - \prod_{l=1}^{nd}U_l\|\\
    &= \|\sum_{l=1}^{nd} U_{nd}...U_{l+1}\delta_l U_{l-1}...U_1 + \cdots \|\\
    &\leq \sum_{l=1}^{nd} \|\delta_l\|+\sum_{l,k=1}^{nd} \|\delta_l\delta_k\| + \cdots \\
    &\in O(nd\epsilon') .
    \end{align}
\ees
Thus the total variation distance between the target BQCA and applied QAOA circuits becomes
\bes\label{eq:TVDbound}
\begin{align} 
    &\sum_{\bm{s}} |P'(\bm{s}) - P(\bm{s})| 
    \leq \text{Tr}[|C'|0\rangle \langle 0 | C'^\dag - C | 0 \rangle \langle 0 | C^\dag|]\\
    &= \text{Tr}[|(\Delta+C)|0\rangle \langle 0 |(\Delta+C)^\dag - C | 0 \rangle \langle 0 | C^\dag|]\\ 
    \label{eq:B3c}
    &= \text{Tr}[|\Delta|0\rangle \langle 0 |\Delta^\dag + \Delta|0\rangle \langle 0 |C^\dag + C|0\rangle \langle 0 |\Delta^\dag|]\\ 
    &\leq \text{Tr}[|\Delta|0\rangle \langle 0 |\Delta^\dag| + |\Delta|0\rangle \langle 0 |C^\dag| + |C|0\rangle \langle 0 |\Delta^\dag|]\\
    &\leq \|\Delta\|^2 \||0\rangle \langle 0 |\|_1 + 2 \|\Delta\| \|C|0\rangle \langle 0 |\|_1\\
    &\leq \|\Delta\|^2 + 2\|\Delta\|\\
    &\in O(nd \epsilon')
    \end{align}
    \ees
where $\bm{s}$ represents a bitstring, the first inequality comes from bounding the total variation distance by the trace-norm distance,\footnote{$\sum_{\bm{s}} |P'(\bm{s}) - P(\bm{s})|\leq \Tr|\rho'-\rho|$, where $P(\bm{s}) = \Tr(\ket{\bm{s}}\!\bra{\bm{s}}\rho)$, with $\rho = C\ket{0}\!\bra{0}C^\dagger$, and the inequality follows since the trace norm distance is the maximum of $\frac{1}{2}\sum_{\bm{s}} |\Tr(E_{\bm{s}}(\rho'-\rho))|$ over all possible generalized measurements $E_{\bm{s}}$, which includes the projective measurement $\ket{\bm{s}}\!\bra{\bm{s}}$~\cite{Lidar-Brun:book}.} the third inequality from properties of the operator and trace norms,\footnote{$|\Tr(AB)|\leq \|A\| \|B^\dagger\|_1$~\cite{Bhatia:book,PhysRevA.78.012308}.} and in the last line we assumed that $\epsilon'<1/(nd)$.

Suppose we wish to approximate a circuit to total variation distance $O(\epsilon)$. Then assuming the worst case bound in Eq.~\eqref{eq:TVDbound}, we must have $\epsilon' \sim \frac{\epsilon}{nd}$. In this case the total runtime of the QAOA circuit is $O(nd\epsilon'^{-4}) = O(n^5 d^5 \epsilon^{-4})$.

\section{Failure of the Euler angles construction to generate single-qubit gates in the V$Z$ model}
\label{app:Euler}

We can define su$(2)$ generators in the V$Z$ model: 
\beq
\tilde{X}=\cos(\alpha)X+\sin(\alpha)Z \ , \ \tilde{Z}=-\sin(\alpha)X+\cos(\alpha)Z .
\eeq
It is simple to check that this pair satisfies the su$(2)$ commutation relations along with $Y$ (e.g., $[\tilde{Z},\tilde{X}]=2iY$, etc.). From here we can construct any SU$(2)$ single-qubit gate using the standard Euler angles construction:
\bes
\begin{align}
& g(\phi,\theta,\psi) = U^z(\phi)U^x(\theta)U^z(\psi) = \\
& \quad \begin{bmatrix} 
 \cos (\theta ) e^{-i (\psi +\phi )} & -i \sin (\theta ) e^{-i (\phi -\psi )} \\
 -i \sin (\theta ) e^{i (\phi -\psi )} & \cos (\theta ) e^{i (\psi +\phi )} \\
\end{bmatrix},
\end{align}
\ees
where the angles take values within the intervals $\theta\in[0,\pi/2]$, $\phi\in[0,\pi]$, $\psi\in[0,2\pi]$ (mod $\pi$), and where $U^x(\varphi) = \exp(-i \varphi \tilde{X})$ and $U^z(\varphi) = \exp(-i \varphi\tilde{Z})$. This approach is certainly feasible for an applied layer in which all qubits undergo the same single-qubit gate. However, it fails when at least one qubit (but not all) in an applied layer is idle, since in the Euler angle construction the only way to generate the identity gate is to choose the angles as
\beq
\{\theta = 0,\psi = 2\pi-\phi \}\ \ \mod \pi .
\eeq
The problem is that fixing $\theta=0$ restricts the ability to generate an arbitrary single-qubit gate on the non-idle qubits. I.e., suppose that the applied layer includes one idle and another non-idle qubit requiring the pure $\tilde{X}$-rotation $g_2(0,\theta,0)=g(\phi,\theta,2\pi-\phi)$ with $\theta>0$. It is not possible to implement both gates without restricting generality within the same time-interval (applied layer) since this limits the allowed values of $\theta$. Namely, for a layer of duration $t$, on the one hand we would need $t\cos\alpha = k\pi$ and $t\sin\alpha = k\pi$ for the idle qubit (using $\theta=k\pi$ with integer $k$ instead of $\theta=0$), but on the other hand we would also need $t\cos\alpha =\theta$ and $t\sin\alpha =\theta$ for the non-idle qubit, thus forcing $\theta$ to be a multiple of $\pi$.

\section{Proof of Eq.~\eqref{eq:psialpha'}}
\label{app:sol}

In this appendix our goal is to show that we can decompose a desired arbitrary single-qubit rotation 
\beq
\label{eq:decomps-g}
g = e^{-i\gamma' \vec r \cdot \vec \sigma} ,
\eeq
where $\vec r = (\sin(2\theta)\cos(2\phi),\ \sin(2\theta)\sin(2\phi),\ \cos(2\theta))$,
as $g = V U V^\dag$, where
\begin{subequations}
\label{eq:decomps}
\begin{align}
\label{eq:decomps-U}
&U = e^{-i\gamma' (\cos(2\alpha)Z + \sin(2\alpha)X)}\\
\label{eq:decomps-V}
&V = e^{-i\alpha'(\cos(2\psi)Z+\sin(2\psi) X)} .
 \end{align}
\end{subequations}
Our second goal is to derive the values of $\alpha'$ and $\psi$ given in Eq.~\eqref{eq:psialpha'}. In Eq.~\eqref{eq:mid_rot} we take $\gamma'=\pi+\gamma$. However, the results of this proof do not depend on specific choice of $\gamma'$, aside from implicit dependence hidden in $\alpha$ or other parameters.

Defining the adjoint group action as $\text{Ad}_{g}(B) \equiv gBg^{-1}$ and using the fact that
\begin{equation}\label{eq:ABA}
    \text{Ad}_{A}(e^B) = e^{ABA^\dag} 
\end{equation}
for any unitary $A$ and operator $B$, and that
\begin{equation}\label{eq:axisrot}
    \text{Ad}_{e^{i\theta Y}}(Z) = \cos(2\theta)Z+\sin(2\theta)X ,
\end{equation}
and similar identities obtained by cycling $X$, $Y$ and $Z$, we may rewrite
\begin{subequations}\label{eq:UV2}
\begin{align}
\label{eq:UV2-1}
U &= \text{Ad}_{e^{i\alpha Y}}(e^{-i\gamma' Z})\\
\label{eq:UV2-2}
V &= \text{Ad}_{e^{i\psi Y}} (e^{-i\alpha' Z}) .
\end{align}
\end{subequations}
Then setting $u = 2\alpha-2\psi$, we see that:
\bes
\begin{align}
    & V U V^\dag  = \text{Ad}_{e^{i\psi Y}}(\text{Ad}_{e^{-i\alpha' Z}}(\text{Ad}_{e^{i(\alpha-\psi)}}(e^{-i\gamma' Z}))) \\
    &= \text{Ad}_{e^{i\psi Y}}(\text{Ad}_{e^{-i\alpha' Z}}(e^{-i\gamma' (\cos(u)Z+\sin(u)X)}))\\
& = \text{Ad}_{e^{i\psi Y}}e^{-i\gamma' (\cos(u)Z+\sin(u)(\cos(2\alpha')X-\sin(2\alpha')Y))}\\
    &= e^{-i\gamma' (v_x X + v_y Y + v_z Z)}
\end{align}
\ees
where in the final line
\begin{subequations}
\begin{align}
\label{}
v_x &= \cos(u)\sin(2\psi)+\sin(u)\cos(2\alpha')\cos(2\psi)\\ 
\label{}
v_y &= -\sin(u)\sin(2\alpha')\\ 
\label{}
v_z &= \cos(u)\cos(2\psi)-\sin(u)\cos(2\alpha')\sin(2\psi) .
\end{align}
\end{subequations}
Using the representation of the desired gate $g$ as in Eq.~\eqref{eq:decomps-g}, this gives a system of three equations:\\
$v_x = r_x$:
\bes
\label{eq:vxterm}
\begin{align}
&\sin(2\theta)\cos(2\phi)\notag\\ 
&= \cos(u)\sin(2\psi)+\sin(u)\cos(2\alpha')\cos(2\psi)\\ 
&= \sin(2\alpha) + \sin(u)\cos(2\psi)(\cos(2\alpha')-1) ,
\end{align}
\ees
and $v_y = r_y$:
\begin{equation}\label{eq:vyterm}
\sin(2\theta)\sin(2\phi) = -\sin(u)\sin(2\alpha') ,
\end{equation}
and $v_z = r_z$:
\bes
\label{eq:vzterm}
\begin{align}
&\cos(2\theta)\notag\\ 
&= \cos(u)\cos(2\psi) - \sin(u)\cos(2\alpha')\sin(2\psi)\\ 
&= \cos(2\alpha) - \sin(u)\sin(2\psi)(\cos(2\alpha')-1) .
\end{align}
\ees
The existence of a solution of these three equations for $\alpha$, $\alpha'$ and $\psi$ proves that $g$ can be written as $VUV^\dagger$ as claimed.

Eq.~\eqref{eq:vyterm} immediately reduces to 
\begin{equation}
    \sin(2\alpha') = -\frac{\sin(2\theta)\sin(2\phi)}{\sin(u)} ,
\end{equation}
while  Eqs.~\eqref{eq:vxterm} and~\eqref{eq:vzterm} can be combined into
\begin{equation}
    \tan(2\psi) = \frac{\cos(2\alpha)-\cos(2\theta)}{\sin(2\theta)\cos(2\phi)-\sin(2\alpha)} .
\end{equation}
Thus, we have the values of $\alpha'$ and $\psi$ given in Eq.~\eqref{eq:psialpha'}, up to a choice of the range of $\tan^{-1}$ and $\sin^{-1}$.

\section{Two qubit coupling decomposition}
\label{app:decomp}

The form of the decomposition in Eq.~\eqref{eq:2qubitdecomp} is motivated by Ref.~\cite{zhang2003}; any $U \in SU(4)$ can be decomposed as:
\be \label{eq:su4decompos}
U = k_1 e^{-i (D_1(X \otimes X) + D_2 (Y \otimes Y) + D_3 (Z \otimes Z))} k_2
\ee
where $k_1 , k_2 \in SU(2) \otimes SU(2)$ and $D_1, D_2, D_3 \in \mathbb{R}$. To write $U_{ij} = e^{-it(a (X_i +X_j) + bZ_i Z_j)}$ in this form we first diagonalize $h_{ij} = a (X_i +X_j) + bZ_i Z_j$ and obtain its eigenvalues ($e_i$) and eigenstates ($|e_i \rangle$). Then we can obtain $U_{ij}$ as $\sum_i e^{-it e_i}|e_i \rangle \langle e_i |$. If we decompose $U_{ij}$ in the Pauli basis, we can show that it has the following form:
\begin{align} \nn
& U_{ij} = P_{00} I + P_{01} X_j+ P_{10} X_i + P_{11} X_i X_j  \\
&+ P_{22} Y_i Y_j + P_{33} Z_i Z_j,
\end{align}
with other terms equal to zero. From this form we see that we do not need to consider the most general form of a single qubit gate for $k_1$ and $k_2$. Instead we start from the following ansatz, which we will show to be sufficient:
\be  
     e^{-i\beta (X_i+X_j)}e^{-iD_1X_i X_j} e^{-iD_2Y_i Y_j}e^{-iD_3 Z_i Z_j}e^{-i\beta (X_i+X_j)} .
\ee
With this ansatz, the problem becomes solving the following equality:
\bes
\label{eq:2qubitdecompapp}
\begin{align}
\label{eq:2qubitdecompapp-a}  
     & e^{-it(a (X_i+X_j) + b Z_i Z_j)}\\
    \label{eq:2qubitdecompapp-b}  
    &= e^{-i\beta (X_i+X_j)}e^{-iD_1X_i X_j} e^{-iD_2Y_i Y_j}e^{-iD_3 Z_i Z_j}e^{-i\beta (X_i+X_j)} 
\end{align}
\ees
This requires solving $16$ coupled equations, each for one element of a $4 \times 4$ matrix. To simplify the task we represent both sides of Eq.~\eqref{eq:2qubitdecompapp} in the ``magic basis" defined in Ref.~\cite{Hill1997} as:
\begin{subequations}
\label{eq:magicbasis}
\begin{align}
\label{eq:magicbasis-1}
|\phi_1 \rangle = \frac{1}{\sqrt{2}} (|00\rangle + |11 \rangle )  \qquad |\phi_2 \rangle = \frac{i}{\sqrt{2}} (|00\rangle - |11\rangle) \\
\label{eq:magicbasis-2}
|\phi_3 \rangle = \frac{1}{\sqrt{2}} (|01\rangle - |10 \rangle )  \qquad |\phi_4 \rangle = \frac{i}{\sqrt{2}} (|01\rangle + |10 \rangle ) .
\end{align}
\end{subequations}
In this basis both matrices are sparse, and can be equated term by term to solve for the given parameters. Also, the non-local part of the right hand side of Eq.~\eqref{eq:2qubitdecompapp-b}, $e^{-iD_1 X_i X_j} e^{-iD_2 Y_i Y_j}e^{-iD_3 Z_i Z_j}$, is diagonal in the magic basis. 
The following matrix changes the basis from the computational basis to the magic basis:
\be
Q = \frac{1}{\sqrt{2}}
\begin{bmatrix} 
1 & i & 0 & 0 \\
0 & 0 & 1 & i \\
0 & 0 & -1 & i \\
1 & -i & 0 & 0
\end{bmatrix} .
\ee
Using $Q$ we can write any matrix $U \in SU(4)$ in the magic basis as:
\be
U_{\text{mag}} = \sum_{i,j} (Q^\dagger U Q)_{ji} |\phi_j \rangle \langle \phi_i |
\ee
The ansatz defined in Eq.~\eqref{eq:2qubitdecompapp-b} takes the following form in this basis:
\be \label{eq:Umagic}
\begin{bmatrix}
 u_{11} & 0 & 0 & u_{14} \\
 0 & u_{22} & 0 & 0 \\
  0 & 0 & u_{33} & 0 \\
  u_{41} & 0 & 0 & u_{44} \\
\end{bmatrix} ,
\ee
where
\begin{subequations}
\label{eq:ansmagic}
\begin{align}
\label{eq:ansmagic11}
    u_{11} &= \frac{1}{2} e^{- i (D_1+D_2+D_3)}( e^{2 i D_3} (\cos(4 \beta )-1) \nn \\ 
    &+e^{2 i D_2} (\cos (4 \beta )+1) ) \\
    \label{eq:ansmagic14}
    u_{14} &= e^{-i D_1} \sin (4 \beta ) \cos(D_2-D_3) \\
    \label{eq:ansmagic22}
    u_{22} &= e^{ i (D_1-D_2-D_3)} \\
    \label{eq:ansmagic33}
    u_{33} &= e^{i (D_1+D_2+D_3)} \\
    \label{eq:ansmagic41}
    u_{41} &= -e^{-i D_1} \sin (4 \beta ) \cos(D_2-D_3) \\
    \label{eq:ansmagic44}
    u_{44} &= \frac{1}{2} e^{-i (D_1+D_2+D_3)} (e^{2 i D_2} (\cos (4 \beta )-1) \nn \\ 
    &+e^{2 i D_3} (\cos (4 \beta )+1)) .
\end{align}
\end{subequations}
The unitary defined in Eq.~\eqref{eq:2qubitdecompapp-a} takes the same form as shown in Eq.~\eqref{eq:Umagic} in the magic basis, such that:
\begin{subequations}
\label{eq:magic}
\begin{align}
\label{eq:magic11}
    u_{11} &= \cos \left(t \sqrt{4 a^2+b^2}\right)-\frac{i b \sin \left(t \sqrt{4 a^2+b^2}\right)}{\sqrt{4 a^2+b^2}} \\
    \label{eq:magic14}
    u_{14} &= \frac{2 a \sin \left(t \sqrt{4 a^2+b^2}\right)}{\sqrt{4 a^2+b^2}}\\
    \label{eq:magic22}
    u_{22} &= e^{-i b t} \\
    \label{eq:ansmagic33}
    u_{33} &= e^{i b t} \\
    \label{eq:magic41}
    u_{41} &= -\frac{2 a \sin \left(t \sqrt{4 a^2+b^2}\right)}{\sqrt{4 a^2+b^2}} \\
    \label{eq:magic44}
    u_{44} &= \cos \left(t \sqrt{4 a^2+b^2}\right)+\frac{i b \sin \left(t \sqrt{4 a^2+b^2}\right)}{\sqrt{4 a^2+b^2}} .
\end{align}
\end{subequations}
Equating Eq.~\eqref{eq:ansmagic} and Eq.~\eqref{eq:magic} line by line, we obtain Eq.~\eqref{eq:2qubitsolns}.

\section{Numerical solutions for the sinc function}
\label{app:numeric}

In this appendix we show that for every value of $C \in [0,\pi]$, there exists a $k \in \{0,1,2,3\}$ such that the equation:

\begin{equation}\label{eq:sincsolve}
   \text{sinc}(\sqrt{x^2+(C+k\pi)^2}) = \text{sinc}(C+k\pi) 
\end{equation}
has a solution for some real $x$ [where $x=2at$ in Eq.~\eqref{eq:sincs-4}]. This equation is transcendental and cannot be solved analytically. However, for a given pair of $(C, k)$, numerical methods can approximately solve Eq.~\eqref{eq:sincsolve} in the domain $x>0$ or determine that no solution exists. We checked the solvability of Eq.~\eqref{eq:sincsolve} for $C \in [0,\pi]$ and $k=0,1,2,3$ to determine the minimum range of $k$ needed to make Eq.~\eqref{eq:sincsolve} solvable for all $C$. The results of this numerical determination are plotted in Fig.~\ref{fig:Cplot}, from which one can see that $k\in \{0,1,2,3\}$ is sufficient.

\begin{figure}[h]
    \centering
    \includegraphics[scale=.4]{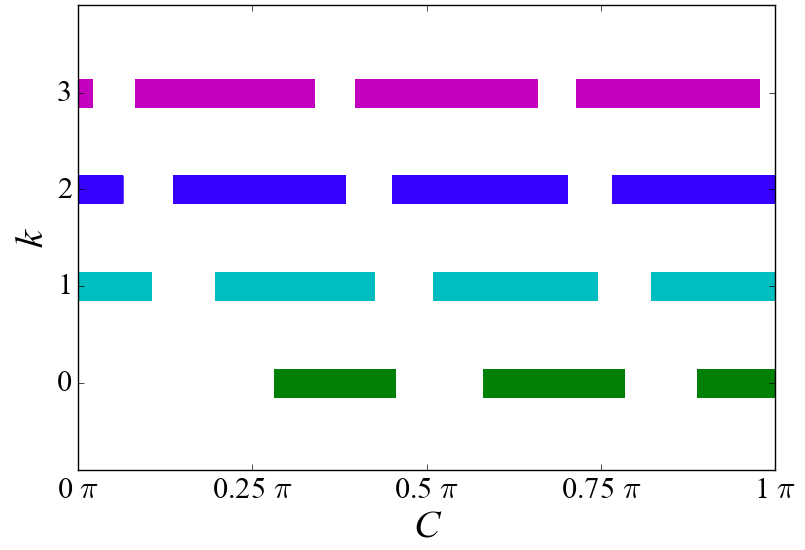}
    \caption{Values of $C$ and $k$ for which Eq.~\eqref{eq:sincsolve} can be solved for $x>0$. Color is included for visibility. For each value of $C$, at least one value of $k$ leads to a solution.}
    \label{fig:Cplot}
\end{figure}

\end{appendix}

%\bibliography{bibli}{}

\begin{thebibliography}{40}%
\makeatletter
\providecommand \@ifxundefined [1]{%
 \@ifx{#1\undefined}
}%
\providecommand \@ifnum [1]{%
 \ifnum #1\expandafter \@firstoftwo
 \else \expandafter \@secondoftwo
 \fi
}%
\providecommand \@ifx [1]{%
 \ifx #1\expandafter \@firstoftwo
 \else \expandafter \@secondoftwo
 \fi
}%
\providecommand \natexlab [1]{#1}%
\providecommand \enquote  [1]{``#1''}%
\providecommand \bibnamefont  [1]{#1}%
\providecommand \bibfnamefont [1]{#1}%
\providecommand \citenamefont [1]{#1}%
\providecommand \href@noop [0]{\@secondoftwo}%
\providecommand \href [0]{\begingroup \@sanitize@url \@href}%
\providecommand \@href[1]{\@@startlink{#1}\@@href}%
\providecommand \@@href[1]{\endgroup#1\@@endlink}%
\providecommand \@sanitize@url [0]{\catcode `\\12\catcode `\$12\catcode
  `\&12\catcode `\#12\catcode `\^12\catcode `\_12\catcode `\%12\relax}%
\providecommand \@@startlink[1]{}%
\providecommand \@@endlink[0]{}%
\providecommand \url  [0]{\begingroup\@sanitize@url \@url }%
\providecommand \@url [1]{\endgroup\@href {#1}{\urlprefix }}%
\providecommand \urlprefix  [0]{URL }%
\providecommand \Eprint [0]{\href }%
\providecommand \doibase [0]{http://dx.doi.org/}%
\providecommand \selectlanguage [0]{\@gobble}%
\providecommand \bibinfo  [0]{\@secondoftwo}%
\providecommand \bibfield  [0]{\@secondoftwo}%
\providecommand \translation [1]{[#1]}%
\providecommand \BibitemOpen [0]{}%
\providecommand \bibitemStop [0]{}%
\providecommand \bibitemNoStop [0]{.\EOS\space}%
\providecommand \EOS [0]{\spacefactor3000\relax}%
\providecommand \BibitemShut  [1]{\csname bibitem#1\endcsname}%
\let\auto@bib@innerbib\@empty
%</preamble>
\bibitem [{\citenamefont {Preskill}(2018)}]{Preskill:2018aa}%
  \BibitemOpen
  \bibfield  {author} {\bibinfo {author} {\bibfnamefont {J.}~\bibnamefont
  {Preskill}},\ }\href {\doibase 10.22331/q-2018-08-06-79} {\bibfield
  {journal} {\bibinfo  {journal} {Quantum}\ }\textbf {\bibinfo {volume} {2}},\
  \bibinfo {pages} {79} (\bibinfo {year} {2018})}\BibitemShut {NoStop}%
\bibitem [{\citenamefont {Deutsch}(1985)}]{Deutsch:85}%
  \BibitemOpen
  \bibfield  {author} {\bibinfo {author} {\bibfnamefont {D.}~\bibnamefont
  {Deutsch}},\ }\bibfield  {booktitle} {\emph {\bibinfo {booktitle}
  {Proceedings of the Royal Society of London. A. Mathematical and Physical
  Sciences}},\ }\href {\doibase 10.1098/rspa.1985.0070} {\bibfield  {journal}
  {\bibinfo  {journal} {Proceedings of the Royal Society of London. A.
  Mathematical and Physical Sciences}\ }\textbf {\bibinfo {volume} {400}},\
  \bibinfo {pages} {97} (\bibinfo {year} {1985})}\BibitemShut {NoStop}%
\bibitem [{\citenamefont {DiVincenzo}(1995)}]{DiVincenzo1995}%
  \BibitemOpen
  \bibfield  {author} {\bibinfo {author} {\bibfnamefont {D.~P.}\ \bibnamefont
  {DiVincenzo}},\ }\href {\doibase 10.1103/physreva.51.1015} {\bibfield
  {journal} {\bibinfo  {journal} {Physical Review A}\ }\textbf {\bibinfo
  {volume} {51}},\ \bibinfo {pages} {1015} (\bibinfo {year}
  {1995})}\BibitemShut {NoStop}%
\bibitem [{\citenamefont {Lloyd}(1995)}]{Lloyd:95}%
  \BibitemOpen
  \bibfield  {author} {\bibinfo {author} {\bibfnamefont {S.}~\bibnamefont
  {Lloyd}},\ }\href {http://link.aps.org/doi/10.1103/PhysRevLett.75.346}
  {\bibfield  {journal} {\bibinfo  {journal} {Physical Review Letters}\
  }\textbf {\bibinfo {volume} {75}},\ \bibinfo {pages} {346} (\bibinfo {year}
  {1995})}\BibitemShut {NoStop}%
\bibitem [{\citenamefont {Lloyd}(2018)}]{lloyd}%
  \BibitemOpen
  \bibfield  {author} {\bibinfo {author} {\bibfnamefont {S.}~\bibnamefont
  {Lloyd}},\ }\href@noop {} {\enquote {\bibinfo {title} {Quantum approximate
  optimization is computationally universal},}\ } (\bibinfo {year} {2018}),\
  \Eprint {http://arxiv.org/abs/1812.11075} {arXiv:1812.11075 [quant-ph]}
  \BibitemShut {NoStop}%
\bibitem [{\citenamefont {Raussendorf}(2005)}]{raussendorf}%
  \BibitemOpen
  \bibfield  {author} {\bibinfo {author} {\bibfnamefont {R.}~\bibnamefont
  {Raussendorf}},\ }\href {http://dx.doi.org/10.1103/PhysRevA.72.052301}
  {\bibfield  {journal} {\bibinfo  {journal} {Physical Review A}\ }\textbf
  {\bibinfo {volume} {72}} (\bibinfo {year} {2005})}\BibitemShut {NoStop}%
\bibitem [{\citenamefont {Preskill}(2012)}]{Preskill:2012aa}%
  \BibitemOpen
  \bibfield  {author} {\bibinfo {author} {\bibfnamefont {J.}~\bibnamefont
  {Preskill}},\ }\href {http://arXiv.org/abs/1203.5813} {\bibfield  {journal}
  {\bibinfo  {journal} {arXiv:1203.5813}\ } (\bibinfo {year}
  {2012})}\BibitemShut {NoStop}%
\bibitem [{\citenamefont {Bremner}\ \emph {et~al.}(2016)\citenamefont
  {Bremner}, \citenamefont {Montanaro},\ and\ \citenamefont
  {Shepherd}}]{bremner1}%
  \BibitemOpen
  \bibfield  {author} {\bibinfo {author} {\bibfnamefont {M.~J.}\ \bibnamefont
  {Bremner}}, \bibinfo {author} {\bibfnamefont {A.}~\bibnamefont {Montanaro}},
  \ and\ \bibinfo {author} {\bibfnamefont {D.~J.}\ \bibnamefont {Shepherd}},\
  }\href {http://dx.doi.org/10.1103/PhysRevLett.117.080501} {\bibfield
  {journal} {\bibinfo  {journal} {Physical Review Letters}\ }\textbf {\bibinfo
  {volume} {117}} (\bibinfo {year} {2016})}\BibitemShut {NoStop}%
\bibitem [{\citenamefont {Farhi}\ and\ \citenamefont {Harrow}(2019)}]{farhi}%
  \BibitemOpen
  \bibfield  {author} {\bibinfo {author} {\bibfnamefont {E.}~\bibnamefont
  {Farhi}}\ and\ \bibinfo {author} {\bibfnamefont {A.~W.}\ \bibnamefont
  {Harrow}},\ }\href@noop {} {\enquote {\bibinfo {title} {Quantum supremacy
  through the quantum approximate optimization algorithm},}\ } (\bibinfo {year}
  {2019}),\ \Eprint {http://arxiv.org/abs/1602.07674} {arXiv:1602.07674
  [quant-ph]} \BibitemShut {NoStop}%
\bibitem [{\citenamefont {Aaronson}\ and\ \citenamefont
  {Chen}(2017)}]{aaronson2016}%
  \BibitemOpen
  \bibfield  {author} {\bibinfo {author} {\bibfnamefont {S.}~\bibnamefont
  {Aaronson}}\ and\ \bibinfo {author} {\bibfnamefont {L.}~\bibnamefont
  {Chen}},\ }in\ \href {https://arxiv.org/abs/1612.05903} {\emph {\bibinfo
  {booktitle} {Proceedings of the 32nd Computational Complexity Conference}}},\
  \bibinfo {series and number} {CCC '17}\ (\bibinfo  {publisher} {Schloss
  Dagstuhl--Leibniz-Zentrum fuer Informatik},\ \bibinfo {address} {Dagstuhl,
  DEU},\ \bibinfo {year} {2017})\BibitemShut {NoStop}%
\bibitem [{\citenamefont {Harrow}\ and\ \citenamefont
  {Montanaro}(2017)}]{Harrow:2017aa}%
  \BibitemOpen
  \bibfield  {author} {\bibinfo {author} {\bibfnamefont {A.~W.}\ \bibnamefont
  {Harrow}}\ and\ \bibinfo {author} {\bibfnamefont {A.}~\bibnamefont
  {Montanaro}},\ }\href {http://dx.doi.org/10.1038/nature23458} {\bibfield
  {journal} {\bibinfo  {journal} {Nature}\ }\textbf {\bibinfo {volume} {549}},\
  \bibinfo {pages} {203 EP } (\bibinfo {year} {2017})}\BibitemShut {NoStop}%
\bibitem [{\citenamefont {Bouland}\ \emph {et~al.}(2018)\citenamefont
  {Bouland}, \citenamefont {Fefferman}, \citenamefont {Nirkhe},\ and\
  \citenamefont {Vazirani}}]{bouland2018}%
  \BibitemOpen
  \bibfield  {author} {\bibinfo {author} {\bibfnamefont {A.}~\bibnamefont
  {Bouland}}, \bibinfo {author} {\bibfnamefont {B.}~\bibnamefont {Fefferman}},
  \bibinfo {author} {\bibfnamefont {C.}~\bibnamefont {Nirkhe}}, \ and\ \bibinfo
  {author} {\bibfnamefont {U.}~\bibnamefont {Vazirani}},\ }\href {\doibase
  10.1038/s41567-018-0318-2} {\bibfield  {journal} {\bibinfo  {journal} {Nature
  Physics}\ }\textbf {\bibinfo {volume} {15}},\ \bibinfo {pages} {159}
  (\bibinfo {year} {2018})}\BibitemShut {NoStop}%
\bibitem [{\citenamefont {Bermejo-Vega}\ \emph {et~al.}(2018)\citenamefont
  {Bermejo-Vega}, \citenamefont {Hangleiter}, \citenamefont {Schwarz},
  \citenamefont {Raussendorf},\ and\ \citenamefont {Eisert}}]{bermejo}%
  \BibitemOpen
  \bibfield  {author} {\bibinfo {author} {\bibfnamefont {J.}~\bibnamefont
  {Bermejo-Vega}}, \bibinfo {author} {\bibfnamefont {D.}~\bibnamefont
  {Hangleiter}}, \bibinfo {author} {\bibfnamefont {M.}~\bibnamefont {Schwarz}},
  \bibinfo {author} {\bibfnamefont {R.}~\bibnamefont {Raussendorf}}, \ and\
  \bibinfo {author} {\bibfnamefont {J.}~\bibnamefont {Eisert}},\ }\href
  {http://dx.doi.org/10.1103/PhysRevX.8.021010} {\bibfield  {journal} {\bibinfo
   {journal} {Physical Review X}\ }\textbf {\bibinfo {volume} {8}} (\bibinfo
  {year} {2018})}\BibitemShut {NoStop}%
\bibitem [{\citenamefont {Haferkamp}\ \emph {et~al.}(2019)\citenamefont
  {Haferkamp}, \citenamefont {Hangleiter}, \citenamefont {Bouland},
  \citenamefont {Fefferman}, \citenamefont {Eisert},\ and\ \citenamefont
  {Bermejo-Vega}}]{haferkamp}%
  \BibitemOpen
  \bibfield  {author} {\bibinfo {author} {\bibfnamefont {J.}~\bibnamefont
  {Haferkamp}}, \bibinfo {author} {\bibfnamefont {D.}~\bibnamefont
  {Hangleiter}}, \bibinfo {author} {\bibfnamefont {A.}~\bibnamefont {Bouland}},
  \bibinfo {author} {\bibfnamefont {B.}~\bibnamefont {Fefferman}}, \bibinfo
  {author} {\bibfnamefont {J.}~\bibnamefont {Eisert}}, \ and\ \bibinfo {author}
  {\bibfnamefont {J.}~\bibnamefont {Bermejo-Vega}},\ }\href@noop {} {\enquote
  {\bibinfo {title} {Closing gaps of a quantum advantage with short-time
  hamiltonian dynamics},}\ } (\bibinfo {year} {2019}),\ \Eprint
  {http://arxiv.org/abs/1908.08069} {arXiv:1908.08069 [quant-ph]} \BibitemShut
  {NoStop}%
\bibitem [{\citenamefont {Arute}\ \emph {et~al.}(2019)\citenamefont {Arute},
  \citenamefont {Arya}, \citenamefont {Babbush}, \citenamefont {Bacon},
  \citenamefont {Bardin}, \citenamefont {Barends}, \citenamefont {Biswas},
  \citenamefont {Boixo}, \citenamefont {Brandao}, \citenamefont {Buell},
  \citenamefont {Burkett}, \citenamefont {Chen}, \citenamefont {Chen},
  \citenamefont {Chiaro}, \citenamefont {Collins}, \citenamefont {Courtney},
  \citenamefont {Dunsworth}, \citenamefont {Farhi}, \citenamefont {Foxen},
  \citenamefont {Fowler}, \citenamefont {Gidney}, \citenamefont {Giustina},
  \citenamefont {Graff}, \citenamefont {Guerin}, \citenamefont {Habegger},
  \citenamefont {Harrigan}, \citenamefont {Hartmann}, \citenamefont {Ho},
  \citenamefont {Hoffmann}, \citenamefont {Huang}, \citenamefont {Humble},
  \citenamefont {Isakov}, \citenamefont {Jeffrey}, \citenamefont {Jiang},
  \citenamefont {Kafri}, \citenamefont {Kechedzhi}, \citenamefont {Kelly},
  \citenamefont {Klimov}, \citenamefont {Knysh}, \citenamefont {Korotkov},
  \citenamefont {Kostritsa}, \citenamefont {Landhuis}, \citenamefont
  {Lindmark}, \citenamefont {Lucero}, \citenamefont {Lyakh}, \citenamefont
  {Mandr{\`a}}, \citenamefont {McClean}, \citenamefont {McEwen}, \citenamefont
  {Megrant}, \citenamefont {Mi}, \citenamefont {Michielsen}, \citenamefont
  {Mohseni}, \citenamefont {Mutus}, \citenamefont {Naaman}, \citenamefont
  {Neeley}, \citenamefont {Neill}, \citenamefont {Niu}, \citenamefont {Ostby},
  \citenamefont {Petukhov}, \citenamefont {Platt}, \citenamefont {Quintana},
  \citenamefont {Rieffel}, \citenamefont {Roushan}, \citenamefont {Rubin},
  \citenamefont {Sank}, \citenamefont {Satzinger}, \citenamefont {Smelyanskiy},
  \citenamefont {Sung}, \citenamefont {Trevithick}, \citenamefont
  {Vainsencher}, \citenamefont {Villalonga}, \citenamefont {White},
  \citenamefont {Yao}, \citenamefont {Yeh}, \citenamefont {Zalcman},
  \citenamefont {Neven},\ and\ \citenamefont {Martinis}}]{Arute:2019aa}%
  \BibitemOpen
  \bibfield  {author} {\bibinfo {author} {\bibfnamefont {F.}~\bibnamefont
  {Arute}}, \bibinfo {author} {\bibfnamefont {K.}~\bibnamefont {Arya}},
  \bibinfo {author} {\bibfnamefont {R.}~\bibnamefont {Babbush}}, \bibinfo
  {author} {\bibfnamefont {D.}~\bibnamefont {Bacon}}, \bibinfo {author}
  {\bibfnamefont {J.~C.}\ \bibnamefont {Bardin}}, \bibinfo {author}
  {\bibfnamefont {R.}~\bibnamefont {Barends}}, \bibinfo {author} {\bibfnamefont
  {R.}~\bibnamefont {Biswas}}, \bibinfo {author} {\bibfnamefont
  {S.}~\bibnamefont {Boixo}}, \bibinfo {author} {\bibfnamefont {F.~G. S.~L.}\
  \bibnamefont {Brandao}}, \bibinfo {author} {\bibfnamefont {D.~A.}\
  \bibnamefont {Buell}}, \bibinfo {author} {\bibfnamefont {B.}~\bibnamefont
  {Burkett}}, \bibinfo {author} {\bibfnamefont {Y.}~\bibnamefont {Chen}},
  \bibinfo {author} {\bibfnamefont {Z.}~\bibnamefont {Chen}}, \bibinfo {author}
  {\bibfnamefont {B.}~\bibnamefont {Chiaro}}, \bibinfo {author} {\bibfnamefont
  {R.}~\bibnamefont {Collins}}, \bibinfo {author} {\bibfnamefont
  {W.}~\bibnamefont {Courtney}}, \bibinfo {author} {\bibfnamefont
  {A.}~\bibnamefont {Dunsworth}}, \bibinfo {author} {\bibfnamefont
  {E.}~\bibnamefont {Farhi}}, \bibinfo {author} {\bibfnamefont
  {B.}~\bibnamefont {Foxen}}, \bibinfo {author} {\bibfnamefont
  {A.}~\bibnamefont {Fowler}}, \bibinfo {author} {\bibfnamefont
  {C.}~\bibnamefont {Gidney}}, \bibinfo {author} {\bibfnamefont
  {M.}~\bibnamefont {Giustina}}, \bibinfo {author} {\bibfnamefont
  {R.}~\bibnamefont {Graff}}, \bibinfo {author} {\bibfnamefont
  {K.}~\bibnamefont {Guerin}}, \bibinfo {author} {\bibfnamefont
  {S.}~\bibnamefont {Habegger}}, \bibinfo {author} {\bibfnamefont {M.~P.}\
  \bibnamefont {Harrigan}}, \bibinfo {author} {\bibfnamefont {M.~J.}\
  \bibnamefont {Hartmann}}, \bibinfo {author} {\bibfnamefont {A.}~\bibnamefont
  {Ho}}, \bibinfo {author} {\bibfnamefont {M.}~\bibnamefont {Hoffmann}},
  \bibinfo {author} {\bibfnamefont {T.}~\bibnamefont {Huang}}, \bibinfo
  {author} {\bibfnamefont {T.~S.}\ \bibnamefont {Humble}}, \bibinfo {author}
  {\bibfnamefont {S.~V.}\ \bibnamefont {Isakov}}, \bibinfo {author}
  {\bibfnamefont {E.}~\bibnamefont {Jeffrey}}, \bibinfo {author} {\bibfnamefont
  {Z.}~\bibnamefont {Jiang}}, \bibinfo {author} {\bibfnamefont
  {D.}~\bibnamefont {Kafri}}, \bibinfo {author} {\bibfnamefont
  {K.}~\bibnamefont {Kechedzhi}}, \bibinfo {author} {\bibfnamefont
  {J.}~\bibnamefont {Kelly}}, \bibinfo {author} {\bibfnamefont {P.~V.}\
  \bibnamefont {Klimov}}, \bibinfo {author} {\bibfnamefont {S.}~\bibnamefont
  {Knysh}}, \bibinfo {author} {\bibfnamefont {A.}~\bibnamefont {Korotkov}},
  \bibinfo {author} {\bibfnamefont {F.}~\bibnamefont {Kostritsa}}, \bibinfo
  {author} {\bibfnamefont {D.}~\bibnamefont {Landhuis}}, \bibinfo {author}
  {\bibfnamefont {M.}~\bibnamefont {Lindmark}}, \bibinfo {author}
  {\bibfnamefont {E.}~\bibnamefont {Lucero}}, \bibinfo {author} {\bibfnamefont
  {D.}~\bibnamefont {Lyakh}}, \bibinfo {author} {\bibfnamefont
  {S.}~\bibnamefont {Mandr{\`a}}}, \bibinfo {author} {\bibfnamefont {J.~R.}\
  \bibnamefont {McClean}}, \bibinfo {author} {\bibfnamefont {M.}~\bibnamefont
  {McEwen}}, \bibinfo {author} {\bibfnamefont {A.}~\bibnamefont {Megrant}},
  \bibinfo {author} {\bibfnamefont {X.}~\bibnamefont {Mi}}, \bibinfo {author}
  {\bibfnamefont {K.}~\bibnamefont {Michielsen}}, \bibinfo {author}
  {\bibfnamefont {M.}~\bibnamefont {Mohseni}}, \bibinfo {author} {\bibfnamefont
  {J.}~\bibnamefont {Mutus}}, \bibinfo {author} {\bibfnamefont
  {O.}~\bibnamefont {Naaman}}, \bibinfo {author} {\bibfnamefont
  {M.}~\bibnamefont {Neeley}}, \bibinfo {author} {\bibfnamefont
  {C.}~\bibnamefont {Neill}}, \bibinfo {author} {\bibfnamefont {M.~Y.}\
  \bibnamefont {Niu}}, \bibinfo {author} {\bibfnamefont {E.}~\bibnamefont
  {Ostby}}, \bibinfo {author} {\bibfnamefont {A.}~\bibnamefont {Petukhov}},
  \bibinfo {author} {\bibfnamefont {J.~C.}\ \bibnamefont {Platt}}, \bibinfo
  {author} {\bibfnamefont {C.}~\bibnamefont {Quintana}}, \bibinfo {author}
  {\bibfnamefont {E.~G.}\ \bibnamefont {Rieffel}}, \bibinfo {author}
  {\bibfnamefont {P.}~\bibnamefont {Roushan}}, \bibinfo {author} {\bibfnamefont
  {N.~C.}\ \bibnamefont {Rubin}}, \bibinfo {author} {\bibfnamefont
  {D.}~\bibnamefont {Sank}}, \bibinfo {author} {\bibfnamefont {K.~J.}\
  \bibnamefont {Satzinger}}, \bibinfo {author} {\bibfnamefont {V.}~\bibnamefont
  {Smelyanskiy}}, \bibinfo {author} {\bibfnamefont {K.~J.}\ \bibnamefont
  {Sung}}, \bibinfo {author} {\bibfnamefont {M.~D.}\ \bibnamefont
  {Trevithick}}, \bibinfo {author} {\bibfnamefont {A.}~\bibnamefont
  {Vainsencher}}, \bibinfo {author} {\bibfnamefont {B.}~\bibnamefont
  {Villalonga}}, \bibinfo {author} {\bibfnamefont {T.}~\bibnamefont {White}},
  \bibinfo {author} {\bibfnamefont {Z.~J.}\ \bibnamefont {Yao}}, \bibinfo
  {author} {\bibfnamefont {P.}~\bibnamefont {Yeh}}, \bibinfo {author}
  {\bibfnamefont {A.}~\bibnamefont {Zalcman}}, \bibinfo {author} {\bibfnamefont
  {H.}~\bibnamefont {Neven}}, \ and\ \bibinfo {author} {\bibfnamefont {J.~M.}\
  \bibnamefont {Martinis}},\ }\href {https://doi.org/10.1038/s41586-019-1666-5}
  {\bibfield  {journal} {\bibinfo  {journal} {Nature}\ }\textbf {\bibinfo
  {volume} {574}},\ \bibinfo {pages} {505} (\bibinfo {year}
  {2019})}\BibitemShut {NoStop}%
\bibitem [{\citenamefont {Aharonov}\ \emph {et~al.}(2009)\citenamefont
  {Aharonov}, \citenamefont {Gottesman}, \citenamefont {Irani},\ and\
  \citenamefont {Kempe}}]{Aharonov:2009tm}%
  \BibitemOpen
  \bibfield  {author} {\bibinfo {author} {\bibfnamefont {D.}~\bibnamefont
  {Aharonov}}, \bibinfo {author} {\bibfnamefont {D.}~\bibnamefont {Gottesman}},
  \bibinfo {author} {\bibfnamefont {S.}~\bibnamefont {Irani}}, \ and\ \bibinfo
  {author} {\bibfnamefont {J.}~\bibnamefont {Kempe}},\ }\href {\doibase
  10.1007/s00220-008-0710-3} {\bibfield  {journal} {\bibinfo  {journal}
  {Communications in Mathematical Physics}\ }\textbf {\bibinfo {volume}
  {287}},\ \bibinfo {pages} {41} (\bibinfo {year} {2009})}\BibitemShut
  {NoStop}%
\bibitem [{\citenamefont {Meier}\ \emph {et~al.}(2015)\citenamefont {Meier},
  \citenamefont {Brierley}, \citenamefont {Kou}, \citenamefont {Girvin},\ and\
  \citenamefont {Glazman}}]{Meier:2015aa}%
  \BibitemOpen
  \bibfield  {author} {\bibinfo {author} {\bibfnamefont {H.}~\bibnamefont
  {Meier}}, \bibinfo {author} {\bibfnamefont {R.~T.}\ \bibnamefont {Brierley}},
  \bibinfo {author} {\bibfnamefont {A.}~\bibnamefont {Kou}}, \bibinfo {author}
  {\bibfnamefont {S.~M.}\ \bibnamefont {Girvin}}, \ and\ \bibinfo {author}
  {\bibfnamefont {L.~I.}\ \bibnamefont {Glazman}},\ }\href
  {https://link.aps.org/doi/10.1103/PhysRevB.92.064516} {\bibfield  {journal}
  {\bibinfo  {journal} {Physical Review B}\ }\textbf {\bibinfo {volume} {92}},\
  \bibinfo {pages} {064516} (\bibinfo {year} {2015})}\BibitemShut {NoStop}%
\bibitem [{\citenamefont {Neill}\ \emph {et~al.}(2018)\citenamefont {Neill},
  \citenamefont {Roushan}, \citenamefont {Kechedzhi}, \citenamefont {Boixo},
  \citenamefont {Isakov}, \citenamefont {Smelyanskiy}, \citenamefont {Megrant},
  \citenamefont {Chiaro}, \citenamefont {Dunsworth}, \citenamefont {Arya},
  \citenamefont {Barends}, \citenamefont {Burkett}, \citenamefont {Chen},
  \citenamefont {Chen}, \citenamefont {Fowler}, \citenamefont {Foxen},
  \citenamefont {Giustina}, \citenamefont {Graff}, \citenamefont {Jeffrey},
  \citenamefont {Huang}, \citenamefont {Kelly}, \citenamefont {Klimov},
  \citenamefont {Lucero}, \citenamefont {Mutus}, \citenamefont {Neeley},
  \citenamefont {Quintana}, \citenamefont {Sank}, \citenamefont {Vainsencher},
  \citenamefont {Wenner}, \citenamefont {White}, \citenamefont {Neven},\ and\
  \citenamefont {Martinis}}]{Neill:2018aa}%
  \BibitemOpen
  \bibfield  {author} {\bibinfo {author} {\bibfnamefont {C.}~\bibnamefont
  {Neill}}, \bibinfo {author} {\bibfnamefont {P.}~\bibnamefont {Roushan}},
  \bibinfo {author} {\bibfnamefont {K.}~\bibnamefont {Kechedzhi}}, \bibinfo
  {author} {\bibfnamefont {S.}~\bibnamefont {Boixo}}, \bibinfo {author}
  {\bibfnamefont {S.~V.}\ \bibnamefont {Isakov}}, \bibinfo {author}
  {\bibfnamefont {V.}~\bibnamefont {Smelyanskiy}}, \bibinfo {author}
  {\bibfnamefont {A.}~\bibnamefont {Megrant}}, \bibinfo {author} {\bibfnamefont
  {B.}~\bibnamefont {Chiaro}}, \bibinfo {author} {\bibfnamefont
  {A.}~\bibnamefont {Dunsworth}}, \bibinfo {author} {\bibfnamefont
  {K.}~\bibnamefont {Arya}}, \bibinfo {author} {\bibfnamefont {R.}~\bibnamefont
  {Barends}}, \bibinfo {author} {\bibfnamefont {B.}~\bibnamefont {Burkett}},
  \bibinfo {author} {\bibfnamefont {Y.}~\bibnamefont {Chen}}, \bibinfo {author}
  {\bibfnamefont {Z.}~\bibnamefont {Chen}}, \bibinfo {author} {\bibfnamefont
  {A.}~\bibnamefont {Fowler}}, \bibinfo {author} {\bibfnamefont
  {B.}~\bibnamefont {Foxen}}, \bibinfo {author} {\bibfnamefont
  {M.}~\bibnamefont {Giustina}}, \bibinfo {author} {\bibfnamefont
  {R.}~\bibnamefont {Graff}}, \bibinfo {author} {\bibfnamefont
  {E.}~\bibnamefont {Jeffrey}}, \bibinfo {author} {\bibfnamefont
  {T.}~\bibnamefont {Huang}}, \bibinfo {author} {\bibfnamefont
  {J.}~\bibnamefont {Kelly}}, \bibinfo {author} {\bibfnamefont
  {P.}~\bibnamefont {Klimov}}, \bibinfo {author} {\bibfnamefont
  {E.}~\bibnamefont {Lucero}}, \bibinfo {author} {\bibfnamefont
  {J.}~\bibnamefont {Mutus}}, \bibinfo {author} {\bibfnamefont
  {M.}~\bibnamefont {Neeley}}, \bibinfo {author} {\bibfnamefont
  {C.}~\bibnamefont {Quintana}}, \bibinfo {author} {\bibfnamefont
  {D.}~\bibnamefont {Sank}}, \bibinfo {author} {\bibfnamefont {A.}~\bibnamefont
  {Vainsencher}}, \bibinfo {author} {\bibfnamefont {J.}~\bibnamefont {Wenner}},
  \bibinfo {author} {\bibfnamefont {T.~C.}\ \bibnamefont {White}}, \bibinfo
  {author} {\bibfnamefont {H.}~\bibnamefont {Neven}}, \ and\ \bibinfo {author}
  {\bibfnamefont {J.~M.}\ \bibnamefont {Martinis}},\ }\href
  {http://science.sciencemag.org/content/360/6385/195.abstract} {\bibfield
  {journal} {\bibinfo  {journal} {Science}\ }\textbf {\bibinfo {volume}
  {360}},\ \bibinfo {pages} {195} (\bibinfo {year} {2018})}\BibitemShut
  {NoStop}%
\bibitem [{\citenamefont {You}\ \emph {et~al.}(2007)\citenamefont {You},
  \citenamefont {Hu}, \citenamefont {Ashhab},\ and\ \citenamefont
  {Nori}}]{You:2007aa}%
  \BibitemOpen
  \bibfield  {author} {\bibinfo {author} {\bibfnamefont {J.~Q.}\ \bibnamefont
  {You}}, \bibinfo {author} {\bibfnamefont {X.}~\bibnamefont {Hu}}, \bibinfo
  {author} {\bibfnamefont {S.}~\bibnamefont {Ashhab}}, \ and\ \bibinfo {author}
  {\bibfnamefont {F.}~\bibnamefont {Nori}},\ }\href {\doibase
  10.1103/PhysRevB.75.140515} {\bibfield  {journal} {\bibinfo  {journal}
  {Physical Review B}\ }\textbf {\bibinfo {volume} {75}},\ \bibinfo {pages}
  {140515} (\bibinfo {year} {2007})}\BibitemShut {NoStop}%
\bibitem [{\citenamefont {Harris}\ \emph {et~al.}(2010)\citenamefont {Harris},
  \citenamefont {Johansson}, \citenamefont {Berkley}, \citenamefont {Johnson},
  \citenamefont {Lanting}, \citenamefont {Han}, \citenamefont {Bunyk},
  \citenamefont {Ladizinsky}, \citenamefont {Oh}, \citenamefont {Perminov},
  \citenamefont {Tolkacheva}, \citenamefont {Uchaikin}, \citenamefont
  {Chapple}, \citenamefont {Enderud}, \citenamefont {Rich}, \citenamefont
  {Thom}, \citenamefont {Wang}, \citenamefont {Wilson},\ and\ \citenamefont
  {Rose}}]{harris_flux_qubit_2010}%
  \BibitemOpen
  \bibfield  {author} {\bibinfo {author} {\bibfnamefont {R.}~\bibnamefont
  {Harris}}, \bibinfo {author} {\bibfnamefont {J.}~\bibnamefont {Johansson}},
  \bibinfo {author} {\bibfnamefont {A.~J.}\ \bibnamefont {Berkley}}, \bibinfo
  {author} {\bibfnamefont {M.~W.}\ \bibnamefont {Johnson}}, \bibinfo {author}
  {\bibfnamefont {T.}~\bibnamefont {Lanting}}, \bibinfo {author} {\bibfnamefont
  {S.}~\bibnamefont {Han}}, \bibinfo {author} {\bibfnamefont {P.}~\bibnamefont
  {Bunyk}}, \bibinfo {author} {\bibfnamefont {E.}~\bibnamefont {Ladizinsky}},
  \bibinfo {author} {\bibfnamefont {T.}~\bibnamefont {Oh}}, \bibinfo {author}
  {\bibfnamefont {I.}~\bibnamefont {Perminov}}, \bibinfo {author}
  {\bibfnamefont {E.}~\bibnamefont {Tolkacheva}}, \bibinfo {author}
  {\bibfnamefont {S.}~\bibnamefont {Uchaikin}}, \bibinfo {author}
  {\bibfnamefont {E.~M.}\ \bibnamefont {Chapple}}, \bibinfo {author}
  {\bibfnamefont {C.}~\bibnamefont {Enderud}}, \bibinfo {author} {\bibfnamefont
  {C.}~\bibnamefont {Rich}}, \bibinfo {author} {\bibfnamefont {M.}~\bibnamefont
  {Thom}}, \bibinfo {author} {\bibfnamefont {J.}~\bibnamefont {Wang}}, \bibinfo
  {author} {\bibfnamefont {B.}~\bibnamefont {Wilson}}, \ and\ \bibinfo {author}
  {\bibfnamefont {G.}~\bibnamefont {Rose}},\ }\href {\doibase
  10.1103/PhysRevB.81.134510} {\bibfield  {journal} {\bibinfo  {journal} {Phys.
  Rev. B}\ }\textbf {\bibinfo {volume} {81}},\ \bibinfo {pages} {134510}
  (\bibinfo {year} {2010})}\BibitemShut {NoStop}%
\bibitem [{\citenamefont {Yan}\ \emph {et~al.}(2016)\citenamefont {Yan},
  \citenamefont {Gustavsson}, \citenamefont {Kamal}, \citenamefont {Birenbaum},
  \citenamefont {Sears}, \citenamefont {Hover}, \citenamefont {Gudmundsen},
  \citenamefont {Rosenberg}, \citenamefont {Samach}, \citenamefont {Weber},
  \citenamefont {Yoder}, \citenamefont {Orlando}, \citenamefont {Clarke},
  \citenamefont {Kerman},\ and\ \citenamefont {Oliver}}]{Yan:2016aa}%
  \BibitemOpen
  \bibfield  {author} {\bibinfo {author} {\bibfnamefont {F.}~\bibnamefont
  {Yan}}, \bibinfo {author} {\bibfnamefont {S.}~\bibnamefont {Gustavsson}},
  \bibinfo {author} {\bibfnamefont {A.}~\bibnamefont {Kamal}}, \bibinfo
  {author} {\bibfnamefont {J.}~\bibnamefont {Birenbaum}}, \bibinfo {author}
  {\bibfnamefont {A.~P.}\ \bibnamefont {Sears}}, \bibinfo {author}
  {\bibfnamefont {D.}~\bibnamefont {Hover}}, \bibinfo {author} {\bibfnamefont
  {T.~J.}\ \bibnamefont {Gudmundsen}}, \bibinfo {author} {\bibfnamefont
  {D.}~\bibnamefont {Rosenberg}}, \bibinfo {author} {\bibfnamefont
  {G.}~\bibnamefont {Samach}}, \bibinfo {author} {\bibfnamefont
  {S.}~\bibnamefont {Weber}}, \bibinfo {author} {\bibfnamefont {J.~L.}\
  \bibnamefont {Yoder}}, \bibinfo {author} {\bibfnamefont {T.~P.}\ \bibnamefont
  {Orlando}}, \bibinfo {author} {\bibfnamefont {J.}~\bibnamefont {Clarke}},
  \bibinfo {author} {\bibfnamefont {A.~J.}\ \bibnamefont {Kerman}}, \ and\
  \bibinfo {author} {\bibfnamefont {W.~D.}\ \bibnamefont {Oliver}},\ }\href
  {http://dx.doi.org/10.1038/ncomms12964} {\bibfield  {journal} {\bibinfo
  {journal} {Nature Communications}\ }\textbf {\bibinfo {volume} {7}},\
  \bibinfo {pages} {12964 EP } (\bibinfo {year} {2016})}\BibitemShut {NoStop}%
\bibitem [{\citenamefont {Grover}\ \emph {et~al.}(2020)\citenamefont {Grover},
  \citenamefont {Basham}, \citenamefont {Marakov}, \citenamefont {Disseler},
  \citenamefont {Hinkey}, \citenamefont {Khalil}, \citenamefont {Stegen},
  \citenamefont {Chamberlin}, \citenamefont {DeGottardi}, \citenamefont
  {Clarke}, \citenamefont {Medford}, \citenamefont {Strand}, \citenamefont
  {Stoutimore}, \citenamefont {Novikov}, \citenamefont {Ferguson},
  \citenamefont {Lidar}, \citenamefont {Zick},\ and\ \citenamefont
  {Przybysz}}]{grover2020fast}%
  \BibitemOpen
  \bibfield  {author} {\bibinfo {author} {\bibfnamefont {J.~A.}\ \bibnamefont
  {Grover}}, \bibinfo {author} {\bibfnamefont {J.~I.}\ \bibnamefont {Basham}},
  \bibinfo {author} {\bibfnamefont {A.}~\bibnamefont {Marakov}}, \bibinfo
  {author} {\bibfnamefont {S.~M.}\ \bibnamefont {Disseler}}, \bibinfo {author}
  {\bibfnamefont {R.~T.}\ \bibnamefont {Hinkey}}, \bibinfo {author}
  {\bibfnamefont {M.}~\bibnamefont {Khalil}}, \bibinfo {author} {\bibfnamefont
  {Z.~A.}\ \bibnamefont {Stegen}}, \bibinfo {author} {\bibfnamefont
  {T.}~\bibnamefont {Chamberlin}}, \bibinfo {author} {\bibfnamefont
  {W.}~\bibnamefont {DeGottardi}}, \bibinfo {author} {\bibfnamefont {D.~J.}\
  \bibnamefont {Clarke}}, \bibinfo {author} {\bibfnamefont {J.~R.}\
  \bibnamefont {Medford}}, \bibinfo {author} {\bibfnamefont {J.~D.}\
  \bibnamefont {Strand}}, \bibinfo {author} {\bibfnamefont {M.~J.~A.}\
  \bibnamefont {Stoutimore}}, \bibinfo {author} {\bibfnamefont
  {S.}~\bibnamefont {Novikov}}, \bibinfo {author} {\bibfnamefont {D.~G.}\
  \bibnamefont {Ferguson}}, \bibinfo {author} {\bibfnamefont {D.}~\bibnamefont
  {Lidar}}, \bibinfo {author} {\bibfnamefont {K.~M.}\ \bibnamefont {Zick}}, \
  and\ \bibinfo {author} {\bibfnamefont {A.~J.}\ \bibnamefont {Przybysz}},\
  }\href {https://link.aps.org/doi/10.1103/PRXQuantum.01.020314} {\bibfield
  {journal} {\bibinfo  {journal} {PRX Quantum}\ }\textbf {\bibinfo {volume}
  {1}},\ \bibinfo {pages} {020314} (\bibinfo {year} {2020})}\BibitemShut
  {NoStop}%
\bibitem [{\citenamefont {Mozgunov}\ and\ \citenamefont
  {Lidar}(2020)}]{mozgunov2020quantum}%
  \BibitemOpen
  \bibfield  {author} {\bibinfo {author} {\bibfnamefont {E.}~\bibnamefont
  {Mozgunov}}\ and\ \bibinfo {author} {\bibfnamefont {D.~A.}\ \bibnamefont
  {Lidar}},\ }\href {https://arxiv.org/abs/2011.08116} {\enquote {\bibinfo
  {title} {Quantum adiabatic theorem for unbounded hamiltonians, with
  applications to superconducting circuits},}\ } (\bibinfo {year} {2020}),\
  \Eprint {http://arxiv.org/abs/2011.08116} {arXiv:2011.08116 [quant-ph]}
  \BibitemShut {NoStop}%
\bibitem [{\citenamefont {Farhi}\ \emph {et~al.}(2014)\citenamefont {Farhi},
  \citenamefont {Goldstone},\ and\ \citenamefont {Gutmann}}]{farhi2014quantum}%
  \BibitemOpen
  \bibfield  {author} {\bibinfo {author} {\bibfnamefont {E.}~\bibnamefont
  {Farhi}}, \bibinfo {author} {\bibfnamefont {J.}~\bibnamefont {Goldstone}}, \
  and\ \bibinfo {author} {\bibfnamefont {S.}~\bibnamefont {Gutmann}},\ }\href
  {https://arxiv.org/abs/1411.4028} {\enquote {\bibinfo {title} {A quantum
  approximate optimization algorithm},}\ } (\bibinfo {year} {2014}),\ \Eprint
  {http://arxiv.org/abs/1411.4028} {arXiv:1411.4028 [quant-ph]} \BibitemShut
  {NoStop}%
\bibitem [{\citenamefont {Benjamin}(2001)}]{Simon}%
  \BibitemOpen
  \bibfield  {author} {\bibinfo {author} {\bibfnamefont {S.~C.}\ \bibnamefont
  {Benjamin}},\ }\href {\doibase 10.1103/PhysRevLett.88.017904} {\bibfield
  {journal} {\bibinfo  {journal} {Phys. Rev. Lett.}\ }\textbf {\bibinfo
  {volume} {88}},\ \bibinfo {pages} {017904} (\bibinfo {year}
  {2001})}\BibitemShut {NoStop}%
\bibitem [{\citenamefont {Morales}\ \emph {et~al.}(2020)\citenamefont
  {Morales}, \citenamefont {Biamonte},\ and\ \citenamefont
  {Zimbor{\'a}s}}]{morales}%
  \BibitemOpen
  \bibfield  {author} {\bibinfo {author} {\bibfnamefont {M.~E.}\ \bibnamefont
  {Morales}}, \bibinfo {author} {\bibfnamefont {J.}~\bibnamefont {Biamonte}}, \
  and\ \bibinfo {author} {\bibfnamefont {Z.}~\bibnamefont {Zimbor{\'a}s}},\
  }\href@noop {} {\bibfield  {journal} {\bibinfo  {journal} {Quantum
  Information Processing}\ }\textbf {\bibinfo {volume} {19}},\ \bibinfo {pages}
  {1} (\bibinfo {year} {2020})}\BibitemShut {NoStop}%
\bibitem [{\citenamefont {Barenco}\ \emph {et~al.}(1995)\citenamefont
  {Barenco}, \citenamefont {Bennett}, \citenamefont {Cleve}, \citenamefont
  {DiVincenzo}, \citenamefont {Margolus}, \citenamefont {Shor}, \citenamefont
  {Sleator}, \citenamefont {Smolin},\ and\ \citenamefont
  {Weinfurter}}]{Barenco1995}%
  \BibitemOpen
  \bibfield  {author} {\bibinfo {author} {\bibfnamefont {A.}~\bibnamefont
  {Barenco}}, \bibinfo {author} {\bibfnamefont {C.~H.}\ \bibnamefont
  {Bennett}}, \bibinfo {author} {\bibfnamefont {R.}~\bibnamefont {Cleve}},
  \bibinfo {author} {\bibfnamefont {D.~P.}\ \bibnamefont {DiVincenzo}},
  \bibinfo {author} {\bibfnamefont {N.}~\bibnamefont {Margolus}}, \bibinfo
  {author} {\bibfnamefont {P.}~\bibnamefont {Shor}}, \bibinfo {author}
  {\bibfnamefont {T.}~\bibnamefont {Sleator}}, \bibinfo {author} {\bibfnamefont
  {J.~A.}\ \bibnamefont {Smolin}}, \ and\ \bibinfo {author} {\bibfnamefont
  {H.}~\bibnamefont {Weinfurter}},\ }\href {\doibase 10.1103/physreva.52.3457}
  {\bibfield  {journal} {\bibinfo  {journal} {Physical Review A}\ }\textbf
  {\bibinfo {volume} {52}},\ \bibinfo {pages} {3457} (\bibinfo {year}
  {1995})}\BibitemShut {NoStop}%
\bibitem [{\citenamefont {Nielsen}\ and\ \citenamefont
  {Chuang}(2010)}]{nielsen2010quantum}%
  \BibitemOpen
  \bibfield  {author} {\bibinfo {author} {\bibfnamefont {M.~A.}\ \bibnamefont
  {Nielsen}}\ and\ \bibinfo {author} {\bibfnamefont {I.~L.}\ \bibnamefont
  {Chuang}},\ }\href@noop {} {\emph {\bibinfo {title} {Quantum computation and
  quantum information}}}\ (\bibinfo  {publisher} {{Cambridge University
  Press}},\ \bibinfo {year} {2010})\BibitemShut {NoStop}%
\bibitem [{\citenamefont {Vizing}(1964)}]{vizing}%
  \BibitemOpen
  \bibfield  {author} {\bibinfo {author} {\bibfnamefont {V.}~\bibnamefont
  {Vizing}},\ }\href@noop {} {\bibfield  {journal} {\bibinfo  {journal}
  {Diskret. Analiz}\ }\textbf {\bibinfo {volume} {3}},\ \bibinfo {pages} {25}
  (\bibinfo {year} {1964})}\BibitemShut {NoStop}%
\bibitem [{\citenamefont {Hangleiter}\ \emph {et~al.}(2018)\citenamefont
  {Hangleiter}, \citenamefont {Bermejo-Vega}, \citenamefont {Schwarz},\ and\
  \citenamefont {Eisert}}]{hangleiter}%
  \BibitemOpen
  \bibfield  {author} {\bibinfo {author} {\bibfnamefont {D.}~\bibnamefont
  {Hangleiter}}, \bibinfo {author} {\bibfnamefont {J.}~\bibnamefont
  {Bermejo-Vega}}, \bibinfo {author} {\bibfnamefont {M.}~\bibnamefont
  {Schwarz}}, \ and\ \bibinfo {author} {\bibfnamefont {J.}~\bibnamefont
  {Eisert}},\ }\href {\doibase 10.22331/q-2018-05-22-65} {\bibfield  {journal}
  {\bibinfo  {journal} {Quantum}\ }\textbf {\bibinfo {volume} {2}},\ \bibinfo
  {pages} {65} (\bibinfo {year} {2018})}\BibitemShut {NoStop}%
\bibitem [{\citenamefont {Zlokapa}\ \emph {et~al.}(2020)\citenamefont
  {Zlokapa}, \citenamefont {Boixo},\ and\ \citenamefont
  {Lidar}}]{Zlokapa:2020}%
  \BibitemOpen
  \bibfield  {author} {\bibinfo {author} {\bibfnamefont {A.}~\bibnamefont
  {Zlokapa}}, \bibinfo {author} {\bibfnamefont {S.}~\bibnamefont {Boixo}}, \
  and\ \bibinfo {author} {\bibfnamefont {D.}~\bibnamefont {Lidar}},\ }\href
  {https://arxiv.org/abs/2005.02464} {\enquote {\bibinfo {title} {Boundaries of
  quantum supremacy via random circuit sampling},}\ } (\bibinfo {year}
  {2020}),\ \Eprint {http://arxiv.org/abs/2005.02464} {arXiv:2005.02464
  [quant-ph]} \BibitemShut {NoStop}%
\bibitem [{\citenamefont {Bremner}\ \emph {et~al.}(2017)\citenamefont
  {Bremner}, \citenamefont {Montanaro},\ and\ \citenamefont
  {Shepherd}}]{bremner2}%
  \BibitemOpen
  \bibfield  {author} {\bibinfo {author} {\bibfnamefont {M.~J.}\ \bibnamefont
  {Bremner}}, \bibinfo {author} {\bibfnamefont {A.}~\bibnamefont {Montanaro}},
  \ and\ \bibinfo {author} {\bibfnamefont {D.~J.}\ \bibnamefont {Shepherd}},\
  }\href {http://dx.doi.org/10.22331/q-2017-04-25-8} {\bibfield  {journal}
  {\bibinfo  {journal} {Quantum}\ }\textbf {\bibinfo {volume} {1}},\ \bibinfo
  {pages} {8} (\bibinfo {year} {2017})}\BibitemShut {NoStop}%
\bibitem [{\citenamefont {Zhong}\ \emph {et~al.}(2020)\citenamefont {Zhong},
  \citenamefont {Wang}, \citenamefont {Deng}, \citenamefont {Chen},
  \citenamefont {Peng}, \citenamefont {Luo}, \citenamefont {Qin}, \citenamefont
  {Wu}, \citenamefont {Ding}, \citenamefont {Hu}, \citenamefont {Hu},
  \citenamefont {Yang}, \citenamefont {Zhang}, \citenamefont {Li},
  \citenamefont {Li}, \citenamefont {Jiang}, \citenamefont {Gan}, \citenamefont
  {Yang}, \citenamefont {You}, \citenamefont {Wang}, \citenamefont {Li},
  \citenamefont {Liu}, \citenamefont {Lu},\ and\ \citenamefont
  {Pan}}]{Zhong:2020aa}%
  \BibitemOpen
  \bibfield  {author} {\bibinfo {author} {\bibfnamefont {H.-S.}\ \bibnamefont
  {Zhong}}, \bibinfo {author} {\bibfnamefont {H.}~\bibnamefont {Wang}},
  \bibinfo {author} {\bibfnamefont {Y.-H.}\ \bibnamefont {Deng}}, \bibinfo
  {author} {\bibfnamefont {M.-C.}\ \bibnamefont {Chen}}, \bibinfo {author}
  {\bibfnamefont {L.-C.}\ \bibnamefont {Peng}}, \bibinfo {author}
  {\bibfnamefont {Y.-H.}\ \bibnamefont {Luo}}, \bibinfo {author} {\bibfnamefont
  {J.}~\bibnamefont {Qin}}, \bibinfo {author} {\bibfnamefont {D.}~\bibnamefont
  {Wu}}, \bibinfo {author} {\bibfnamefont {X.}~\bibnamefont {Ding}}, \bibinfo
  {author} {\bibfnamefont {Y.}~\bibnamefont {Hu}}, \bibinfo {author}
  {\bibfnamefont {P.}~\bibnamefont {Hu}}, \bibinfo {author} {\bibfnamefont
  {X.-Y.}\ \bibnamefont {Yang}}, \bibinfo {author} {\bibfnamefont {W.-J.}\
  \bibnamefont {Zhang}}, \bibinfo {author} {\bibfnamefont {H.}~\bibnamefont
  {Li}}, \bibinfo {author} {\bibfnamefont {Y.}~\bibnamefont {Li}}, \bibinfo
  {author} {\bibfnamefont {X.}~\bibnamefont {Jiang}}, \bibinfo {author}
  {\bibfnamefont {L.}~\bibnamefont {Gan}}, \bibinfo {author} {\bibfnamefont
  {G.}~\bibnamefont {Yang}}, \bibinfo {author} {\bibfnamefont {L.}~\bibnamefont
  {You}}, \bibinfo {author} {\bibfnamefont {Z.}~\bibnamefont {Wang}}, \bibinfo
  {author} {\bibfnamefont {L.}~\bibnamefont {Li}}, \bibinfo {author}
  {\bibfnamefont {N.-L.}\ \bibnamefont {Liu}}, \bibinfo {author} {\bibfnamefont
  {C.-Y.}\ \bibnamefont {Lu}}, \ and\ \bibinfo {author} {\bibfnamefont {J.-W.}\
  \bibnamefont {Pan}},\ }\href {\doibase 10.1126/science.abe8770} {\bibfield
  {journal} {\bibinfo  {journal} {Science}\ ,\ \bibinfo {pages} {eabe8770}}
  (\bibinfo {year} {2020})}\BibitemShut {NoStop}%
\bibitem [{\citenamefont {Toda}(1991)}]{toda1991}%
  \BibitemOpen
  \bibfield  {author} {\bibinfo {author} {\bibfnamefont {S.}~\bibnamefont
  {Toda}},\ }\href {\doibase 10.1137/0220053} {\bibfield  {journal} {\bibinfo
  {journal} {SIAM J. Comput.}\ }\textbf {\bibinfo {volume} {20}},\ \bibinfo
  {pages} {865} (\bibinfo {year} {1991})}\BibitemShut {NoStop}%
\bibitem [{\citenamefont {Szegedy}(2019)}]{szegedy2019qaoa}%
  \BibitemOpen
  \bibfield  {author} {\bibinfo {author} {\bibfnamefont {M.}~\bibnamefont
  {Szegedy}},\ }\href {https://arxiv.org/abs/1912.12277} {\enquote {\bibinfo
  {title} {{What do QAOA energies reveal about graphs?}}}\ } (\bibinfo {year}
  {2019}),\ \Eprint {http://arxiv.org/abs/1912.12277} {arXiv:1912.12277
  [quant-ph]} \BibitemShut {NoStop}%
\bibitem [{\citenamefont {Lidar}\ and\ \citenamefont
  {Brun}(2013)}]{Lidar-Brun:book}%
  \BibitemOpen
  \bibinfo {editor} {\bibfnamefont {D.}~\bibnamefont {Lidar}}\ and\ \bibinfo
  {editor} {\bibfnamefont {T.}~\bibnamefont {Brun}},\ eds.,\ \href
  {http://www.cambridge.org/9780521897877} {\emph {\bibinfo {title} {Quantum
  Error Correction}}}\ (\bibinfo  {publisher} {Cambridge University Press},\
  \bibinfo {address} {{Cambridge, UK}},\ \bibinfo {year} {2013})\BibitemShut
  {NoStop}%
\bibitem [{\citenamefont {{R. Bhatia}}(1997)}]{Bhatia:book}%
  \BibitemOpen
  \bibfield  {author} {\bibinfo {author} {\bibnamefont {{R. Bhatia}}},\
  }\href@noop {} {\emph {\bibinfo {title} {{Matrix Analysis}}}},\ \bibinfo
  {series} {{Graduate Texts in Mathematics}}\ No.\ \bibinfo {number} {169}\
  (\bibinfo  {publisher} {Springer-Verlag},\ \bibinfo {address} {{New York}},\
  \bibinfo {year} {1997})\BibitemShut {NoStop}%
\bibitem [{\citenamefont {Lidar}\ \emph {et~al.}(2008)\citenamefont {Lidar},
  \citenamefont {Zanardi},\ and\ \citenamefont
  {Khodjasteh}}]{PhysRevA.78.012308}%
  \BibitemOpen
  \bibfield  {author} {\bibinfo {author} {\bibfnamefont {D.~A.}\ \bibnamefont
  {Lidar}}, \bibinfo {author} {\bibfnamefont {P.}~\bibnamefont {Zanardi}}, \
  and\ \bibinfo {author} {\bibfnamefont {K.}~\bibnamefont {Khodjasteh}},\
  }\href {\doibase 10.1103/PhysRevA.78.012308} {\bibfield  {journal} {\bibinfo
  {journal} {Phys. Rev. A}\ }\textbf {\bibinfo {volume} {78}},\ \bibinfo
  {pages} {012308} (\bibinfo {year} {2008})}\BibitemShut {NoStop}%
\bibitem [{\citenamefont {Zhang}\ \emph {et~al.}(2003)\citenamefont {Zhang},
  \citenamefont {Vala}, \citenamefont {Sastry},\ and\ \citenamefont
  {Whaley}}]{zhang2003}%
  \BibitemOpen
  \bibfield  {author} {\bibinfo {author} {\bibfnamefont {J.}~\bibnamefont
  {Zhang}}, \bibinfo {author} {\bibfnamefont {J.}~\bibnamefont {Vala}},
  \bibinfo {author} {\bibfnamefont {S.}~\bibnamefont {Sastry}}, \ and\ \bibinfo
  {author} {\bibfnamefont {K.~B.}\ \bibnamefont {Whaley}},\ }\href
  {https://journals.aps.org/pra/pdf/10.1103/PhysRevA.67.042313} {\bibfield
  {journal} {\bibinfo  {journal} {Physical Review A}\ }\textbf {\bibinfo
  {volume} {67}} (\bibinfo {year} {2003})}\BibitemShut {NoStop}%
\bibitem [{\citenamefont {Hill}\ and\ \citenamefont
  {Wootters}(1997)}]{Hill1997}%
  \BibitemOpen
  \bibfield  {author} {\bibinfo {author} {\bibfnamefont {S.}~\bibnamefont
  {Hill}}\ and\ \bibinfo {author} {\bibfnamefont {W.~K.}\ \bibnamefont
  {Wootters}},\ }\href {\doibase 10.1103/physrevlett.78.5022} {\bibfield
  {journal} {\bibinfo  {journal} {Physical Review Letters}\ }\textbf {\bibinfo
  {volume} {78}},\ \bibinfo {pages} {5022} (\bibinfo {year}
  {1997})}\BibitemShut {NoStop}%
\end{thebibliography}
%\bibliographystyle{apsrev4-1}


\end{document}